\documentclass[a4paper,USenglish,thm-restate]{lipics-v2019}
\pdfoutput=1

\newcommand{\noisyBinarySearch}{\ensuremath{\mathsf{NoisyBinarySearch}}\xspace}

\usepackage[T1]{fontenc}
\usepackage{inputenc}
\usepackage{soul} 
\usepackage{phfqit}
\usepackage{proba} 
\usepackage{amssymb}
\usepackage{amsmath}
\usepackage{amsthm}
\usepackage{amsfonts}
\usepackage{mathtools}
\usepackage{xspace}
\usepackage{bm}
\usepackage{nicefrac}
\usepackage{commath}
\usepackage{bbm}
\usepackage{boxedminipage}
\usepackage{xparse}
\usepackage[x11names]{xcolor}
\usepackage{float}
\usepackage{multirow}
\usepackage{graphicx}
\usepackage{caption}
\usepackage{subcaption}
\usepackage{ifthen}
\usepackage{algpseudocode}
\usepackage{colortbl} 
\usepackage[nameinlink,capitalize]{cleveref}

\usepackage{hyperref}
    \hypersetup{ colorlinks=true, linkcolor=blue, filecolor=magenta, urlcolor=red,
    citecolor=blue,}
\usepackage[utf8]{inputenc}
\usepackage{environ}
\usepackage{tabu}
\usepackage{mdframed}
\usepackage{algorithm}
\usepackage{cleveref}

\usepackage{amsthm,thmtools}%

\newcommand{\etal}{\text{et al.}\xspace}

\DeclareMathOperator{\polylog}{polylog}
\newcommand{\negl}{\ensuremath{\mathsf{negl}}\xspace}

\newcommand{\zo}{\ensuremath{{\{0,1\}}}\xspace}

\newcommand{\tuple}[1]{\ensuremath{\left(#1\right)}}
\newcommand{\eps}{\ensuremath{\epsilon}}
\newcommand{\defeq}[0]{\ensuremath{{\;\vcentcolon=\;}}\xspace}

\NewEnviron{problem}[1]{%
	\begin{center}\fbox{\parbox{6in}{%
				{\centering\scshape #1\par}%
				\parskip=1ex
				\everypar{\hangindent=1em}%
				\BODY
}}\end{center}}

\newcommand{\superscript}[1]{\ensuremath{^{\mbox{\tiny{\textit{#1}}}}}\xspace}
\def \th {\superscript{th}}     
\def \etal{\,{\it et~al.}\,}
\newcommand{\ignore}[1]{{}}

\newcommand{\Goody}{\ensuremath{\mathsf{Good}}}

\newcommand{\HAM}{\ensuremath{\mathsf{HAM}}}

\newcommand{\Enc}{\ensuremath{\mathsf{Enc}}}
\newcommand{\Dec}{\ensuremath{\mathsf{Dec}}}

\providecommand{\email}[1]{\href{mailto:#1}{\nolinkurl{#1}\xspace}}

\newcommand{\itn}[1]{^{(#1)}}

\newcommand{\tp}{\tuple}
\newcommand{\ED}{\ensuremath{\mathsf{ED}}\xspace}

\newcommand{\inn}{\mathsf{in}\xspace}

\def\*#1{\mathbf{#1}}
\def\+#1{\mathcal{#1}}

\newcommand{\bbN}{\ensuremath{\mathbb{N}}\xspace}
\newcommand{\bbZ}{\ensuremath{\mathbb{Z}}\xspace}
\newcommand{\p}[1]{\ensuremath{^{(#1)}}\xspace}
\newcommand{\buff}{\ensuremath{\mathsf{b}}\xspace}
\newcommand{\tY}{\ensuremath{\tilde{Y}}\xspace}

\newcommand{\bbs}{\textproc{Block-Decode}\xspace}
\newcommand{\buffind}{\textproc{Buff-Find}\xspace}

\DeclareMathOperator{\dist}{dist}

\algnewcommand{\IIf}[1]{\State\algorithmicif\ #1\ \algorithmicthen}
\algnewcommand{\EndIIf}{\unskip\ \algorithmicend\ \algorithmicif}

\usepackage{microtype}

\title{Locally Decodable/Correctable Codes for Insertions and Deletions} 

\titlerunning{Local InsDel Codes} 

\author{Alexander R. Block}{Purdue University, USA }{block9@purdue.edu}{}{Supported by NSF CCF-1910659.}

\author{Jeremiah Blocki}{Purdue University, USA}{jblocki@purdue.edu}{}{Supported by NSF CCF-1910659, CNS-1755708, CNS-1704587 and CNS-1931443}

\author{Elena Grigorescu}{Purdue University, USA}{elena-g@purdue.edu}{}{Supported by NSF CCF-1910659 and NSF CCF-1910411.}

\author{Shubhang Kulkarni\footnote{Work done while at Purdue University, USA.}}{University of Illinois Urbana-Champaign, USA}{smkulka2@illinois.edu}{}{}

\author{Minshen Zhu}{Purdue University, USA }{zhu628@purdue.edu}{}{Supported by NSF CCF-1910659.}

\authorrunning{A. R. Block, J. Blocki, E. Grigorescu, S. Kulkarni, M. Zhu} 

\Copyright{Alexander R. Block and Jeremiah Blocki and Elena Grigorescu and Shubhang Kulkarni and Minshen Zhu} 



\category{} 

\relatedversion{This is the full version of \cite{BBGKZ20}} 

\supplement{}

\nolinenumbers 

\hideLIPIcs  

\begin{document}

\maketitle

\begin{abstract}

Recent efforts in coding theory have focused on building codes for insertions and deletions, called insdel codes, with optimal trade-offs between their redundancy and their error-correction capabilities, as well as {\em efficient} encoding and decoding algorithms.


In many applications, polynomial running time may still be prohibitively expensive, which has motivated the study of codes with {\em super-efficient} decoding algorithms. These have led to 
the well-studied notions of Locally Decodable Codes (LDCs) and Locally Correctable Codes (LCCs).  Inspired by these notions, Ostrovsky and Paskin-Cherniavsky (Information Theoretic Security, 2015) generalized Hamming LDCs to insertions and deletions. To the best of our knowledge, these are the only known results  that study the analogues of Hamming LDCs in channels performing insertions and deletions.

Here we continue the study of insdel codes that admit local algorithms. Specifically, we reprove the results of Ostrovsky and Paskin-Cherniavsky for insdel LDCs using a different set of techniques. We also observe that the techniques extend to constructions of LCCs. Specifically, we obtain insdel LDCs and LCCs from their Hamming LDCs and LCCs analogues, respectively. The rate and error-correction capability blow up only by a constant factor, while the query complexity blows up by a poly log factor in the block length.

Since insdel locally decodable/correctble codes are scarcely studied in the literature, we believe our results and techniques may lead to further research. In particular, we conjecture that constant-query insdel LDCs/LCCs do not exist.
\end{abstract}

\newpage
\section{Introduction}

Building error-correcting codes that can recover from insertions and deletions (a.k.a. ``insdel codes'') has been a central theme in recent advances in coding theory  \cite{Levenshtein_SPD66,Kiwi_expectedlength, guruswami2017deletion,HaeuplerS17,guruswami2018polynomial,GuruswamiL18,HaeuplerSS18,HaeuplerS18,BrakensiekGZ18, ChengJLW18, ChengHLSW19, ChengJ0W19, GuruswamiL19,HaeuplerRS19,Haeupler19, ChengGHL20, cheng2020efficient, LiuTX20,GuruswamiHS20}. Insdel codes are generalizations of Hamming codes, in which the corruptions may be viewed as deleting symbols and then inserting other symbols at the deleted locations.

An insdel code is described by an encoding function $E:\Sigma^k\rightarrow{\Sigma^n}$, which encodes every message of length $k$ into a codeword of block length $n$. The rate of the code is the ratio $\frac{k}{n}$.
Classically, a decoding function $D:\Sigma^*\rightarrow{\Sigma^k}$ takes as input a string $w$ obtained from some $E(m)$ after $\delta n$ insertions and deletions and satisfies $D(w)=m$. 
A fundamental research direction is building codes with high communication rate $\frac{k}{n}$, that are robust against a large $\delta$ fraction of insertions and deletions, which also admit {\em efficient} encoding and decoding algorithms. It is only recently that efficient insdel codes with asymptotically good rate and error-correction parameters have been well-understood \cite{HaeuplerRS19, HaeuplerS18, Haeupler19, LiuTX20, GuruswamiHS20}.

In modern applications, polynomial-time decoding may still be prohibitively expensive when working with large data, and instead {\em super-efficient} codes are even more desirable. 
Such codes admit very fast decoding algorithms that query only few  locations into the received word to recover portions of the data. 
Ostrovsky and Paskin-Cherniavsky \cite{Ostrovsky-InsdelLDC-Compiler} defined the notion of Locally Decodable Insdel Codes,\footnote{In \cite{Ostrovsky-InsdelLDC-Compiler}, they are named {\em Locally Decodable Codes for Edit Distance}. 
} inspired by the notion of Locally Decodable Codes (LDCs) for Hamming errors \cite{KatzT00,SudanTV99}. 
A code defined by an encoding $E:\Sigma^k\rightarrow{\Sigma^n}$ is a $q$-query {Locally Decodable Insdel Code}  (Insdel LDC) if there exists a randomized algorithm $\mathcal D$, such that: (1) for each $i\in [k]$ and message $m\in \Sigma^k$,  $\mathcal D$ can probabilistically recover $m_i$, given query access to a word $w\in \Sigma^*$, which was obtained from $E(m)$ corrupted by $\delta$ fraction of insertions and deletions; and  (2)  $\mathcal D$ makes only $q$ queries into $w$. 
The number of queries $q$ is called the {\em locality} of the code.

The rate, error-correcting capability, and locality of the code  are opposing design features, and optimizing all of them at the same time is impossible. For example, every $2$-query LDCs for Hamming errors must have vanishing rate \cite{KerenidisW04}. While progress in understanding these trade-offs for Hamming errors has spanned several decades  \cite{KerenidisW04, Yekhanin08, Yekhanin12, DvirGY11, Efremenko12, KoppartyMRS17}
(see surveys by Yekhanin \cite{Yekhanin12} and by Kopparty and Saraf \cite{KoppartyS16}), in contrast, the literature on the same trade-offs for the more general insdel codes is scarce. Namely, besides the results of \cite{Ostrovsky-InsdelLDC-Compiler}, to the best of our knowledge, only Haeupler and Shahrasbi \cite{HaeuplerS18} consider the notion of locality in building synchronization strings, which are important components of optimal insdel codes.

The results of \cite{Ostrovsky-InsdelLDC-Compiler} provide a direct reduction from classical Hamming error LDCs to insdel LDCs, which preserves the rate of the code and error-correction
capabilities up to constant factors, and whose locality grows only by a polylogarithmic factor in the block length.

In this paper we revisit the results of Ostrovsky and Paskin-Cherniavsky \cite{Ostrovsky-InsdelLDC-Compiler} and provide an alternate proof, using different combinatorial techniques.
We also observe that these results extend  to building Locally Correctable Insdel Codes (Insdel LCCs) from Locally Correctable Codes (LCCs) for Hamming errors. LCCs are a variant of LDCs, in which the decoder is tasked to locally correct every entry of the encoded message, namely $E(m)_i$, instead of the entries of the message itself. If the message $m$ is part of the encoding $E(m)$, then an LCC is also an LDC. In particular, all linear LCCs (i.e, whose codewords form a vector space) are also LDCs.

\begin{theorem}\label{thm:general}
If there exist $q$-query LDCs/LCCs with encoding $E:\Sigma^k\rightarrow{\Sigma^n}$, that can correct from $\delta$-fraction of Hamming errors, then there exist binary $q\cdot \polylog(n)$-query Insdel LDCs/LCCs  with codeword length $\Theta(n\log|\Sigma|)$, 
that can correct from  $\Theta(\delta)$-fraction of insertions and deletions.
\end{theorem}
We emphasize that the resulting LDC/LCC of \cref{thm:general} is a {\em binary} code, even if the input LDC/LCC is over some higher alphabet $\Sigma$.

%






Classical constructions of LDCs/LCCs for Hamming errors fall into three query-complexity regimes. In the constant-query regime, the best known results are based on matching-vector codes, and give encodings that map $k$ symbols into $\exp( \exp(\sqrt{\log k \log\log k}))$ symbols \cite{Yekhanin08,DvirGY11,Efremenko12}. Since the best lower bounds are only quadratic \cite{Woodruff12}, for all we know so far, it is possible that there exist constant-query complexity LDCs with polynomial block length.
In the $\polylog k$-query regime, Reed-Muller codes are examples of $\log^c k$-query LDCs/LCCs of block length $k^{1+\frac{1}{c-1}+o(1)}$ for some $c>0$ (e.g., see \cite{Yekhanin12}). Finally, there exist sub-polynomial (but super logarithmic)-query complexity LDCs/LCCs with constant rate \cite{KoppartyMRS17}. These relatively recent developments improved upon the previous constant rate codes in the 
 $n^{\epsilon}$-query regime achieved by Reed-Muller codes, and later by more efficient constructions (e.g. \cite{KoppartySY14}). 



Given that our reduction achieves $\polylog n$-query complexity blow-up, the results above in conjunction with \cref{thm:general} give us the following asymptotic results.
\begin{corollary}
There exist $\polylog(k)$-query  Insdel LDCs/LCCs encoding $k$ symbols into $o(k^2)$ symbols, that can correct a constant fraction of insertions and deletions.
\end{corollary}
\begin{corollary}
There exist  $(\log k)^{O(\log \log k)}$-query Insdel LDCs/LCCs with constant rate, that can correct from a constant fraction of insertions and deletions. 
\end{corollary}
Our results, similarly to those in \cite{Ostrovsky-InsdelLDC-Compiler}, do not have implications in the constant-query regime. We conjecture that there do not exist constant-query LDCs/LCCs, regardless of their rate.
Since achieving locality against insertions and deletions appears to be a difficult task,  and the area is in its infancy, we believe our results and techniques may motivate further research.

\subsection{Overview of Techniques}
\noindent{\bf Searching in a Nearly Sorted Array.} To build intuition for our local decoding algorithm we consider the following simpler problem: We are given a nearly sorted array $A$ of $n$ distinct elements. 
By nearly sorted we mean that there is another sorted array $A'$ such that $A'[i] = A[i]$ on all but $n'$ indices. 
Given an input $x$ we would like to quickly find $x$ in the original array. 
In the worst case this would require time at least $\Omega(n')$ so we relax the requirement that we always find $x$ to say that there are at most  $cn'$ items that we fail to find $x$ for some constant $c > 0$. 

To design our noisy binary search algorithm that meets these requirement we borrow a notion of local goodness used in the design and analysis of depth-robust graphs---a combinatorial object that has found many applications in cryptography \cite{EGS75,CCS:AlwBloHar17,EC:AlwBloPie18}. 
In particular, fixing $A$ and $A'$ (sorted) we say that an index $j$ is corrupted if $A[j] \neq A'[j]$. 
We say that an index $i$ is $\theta$-locally good if for any $r \geq 0$ at most $\theta$ fraction of the indices $ j \in [i,\ldots, i+r]$ are corrupted and at most $\theta$ fraction of the indices in $[i-r,i]$ are corrupted. 
If at most $n'$ indices are corrupted then one can prove that at least  $n-2n'/\theta$ indices are $\theta$-locally good~\cite{EGS75}. 

As long as the constant $\theta$ is suitably small we can design an efficient randomize search procedure which, with high probability, correctly locates $x$ whenever $x=A[i]$, provided that the unknown index $i$ is $\theta$-locally good. 
Intuitively, suppose we have already narrowed down our search to the smaller range $I= [i_0, i_1]$. 
The rank of $x=A[i]$ in $A'[i_0],\ldots, A'[i_1]$ is exactly $i-i_0+1$ since $A[i]$ is uncorrupted and the rank of $x$ in $A[i_0],\ldots, A[i_1]$ can change by at most $\pm \theta (i-i_0+1)$ --- at most $\theta (i_1-i_0+1)$ indices $j' \in [i_0,i_1]$ can be corrupted since $i \in [i_0,i_1]$ is $\theta$-locally good. 
Now suppose that we sample $t=\polylog(n)$ indices $j_1,\ldots,j_t \in [i_0, i_1]$ and select the median $y_{med}$ of $A[j_1],\ldots, A[j_t]$. 
With high probability the rank $r$ of $y_{med}$ in $A[j_1],\ldots, A[j_t]$ will be close to $(i_1-i_0+1)/2$; i.e., $|r-(i_1-i_0+1)/2| \leq \delta (i_1-i_0+1)$ for some arbitrarily constant $\delta$ which may depend on the number of samples $t$. 
Thus, for suitable constants $\theta$ and $\delta$ whenever $x > y_{med}$ (resp. $x < y_{med}$) we can safely conclude that $i > i_0 + (i_1-i_0+1)/8$ (resp. $i < i_1 - (i_1-i_0+1)/8$) and search in the smaller interval $I' = [i_0+(i_1-i_0+1)/8, i_1]$ (resp. $I' = [i_0, i_1-(i_1-i_0+1)/8]$). 
In both cases the size of the search space is reduced by a constant multiplicative factor so the procedure will terminate after $O(\log n)$ rounds making $O(t \log n)$ queries. 
At its core our local decoding algorithm relies on a very similar idea. \mbox{}\\

\noindent{\bf Encoding.}
Our encoder builds off of the known techniques of concatenation codes.
First, a message $x$ is encoded via the outer code to obtain some (intermediate) encoding $y$.
We then partition $y$ into some number $k$ blocks $y = y_1 \circ \cdots \circ y_k$ and append each block $y_i$ with index $i$ to obtain $y_i \circ i$.
Each $y_i\circ i$ is then encoded with the inner encoder to obtain some $d_i$.
Then each $d_i$ is prepended and appended with a run of $0$s (i.e., buffers), to obtain $c_i$.
The encoder then outputs $c = c_1 \circ \cdots \circ c_k$ as the final codeword. 
For our inner encoder, we in fact use the Schulman-Zuckerman (SZ) \cite{SchZuc99} edit distance code. \mbox{}\\

\noindent{\bf Decoding.}
Given oracle access to some corrupted codeword $c'$, on input index $i$, the decoder simulates the outer decoder and must answer the outer decoder oracle queries.
The decoder uses the inner decoder to answer these queries.
However, there are two major challenges: (1) Unlike the Hamming-type errors, even only a few insertions and deletions make it difficult for the decoder to know where to probe; and (2) The boundaries between blocks can be ambiguous in the presence of insdel errors. 
We overcome these challenges via a variant of binary search, which we name \textsf{NoisyBinarySearch}, together with a buffer detection algorithm, and make use of a block decomposition of the corrupted codeword to facilitate the analysis. \mbox{}\\

\noindent{\bf Analysis.}
The analyses of the binary search and the buffer detection algorithms are based on the notion of ``good blocks'' and ``locally good blocks'', which are natural extensions of the notion of $\theta$-locally good indices discussed above.
Recall that our encoder outputs a final codeword that is a concatenation of $k$ smaller codeword ``blocks''; namely $\Enc(x) = c_1 \circ \cdots \circ c_k$.
Suppose $c'$ is the corrupted codeword obtained by corrupting $c$ with $\delta$-fraction of insertion-deletion errors, and suppose we have a method of partitioning $c'$ into $k$ blocks $c'_1\circ \cdots \circ c'_{k}$.
Then we say that block $c'_j$ is a $\gamma$-good block if it is within $\gamma$-fractional edit distance to the {\em uncorrupted} block $c_j$.
Moreover, $c'_j$ is $(\theta,\gamma)$-locally good if at least $(1-\theta)$ fraction of the blocks in every neighborhood around $c'_j$ are $\gamma$-good and if the total number of corruptions in every neighborhood is bounded.
Here $\theta$ and $\gamma$ are suitably chosen constants.
Both notions of good and locally good blocks are necessary to the success of our binary search algorithm \textsf{NoisyBinarySearch}. 

The goal of \textsf{NoisyBinarySearch} is to locate a block with a given index $j$, and the idea is to decode the corrupted codeword at random positions to get a list of decoded indices (recall that the index of each block is appended to it). 
Since a large fraction of blocks are $\gamma$-good blocks, the sampled indices induce a new search interval for the next iteration. 
In order to apply this argument recursively, we need that the error density of the search interval does not increase in each iteration. 
Locally good blocks provide precisely this property. \mbox{}\\

\noindent{\bf Comparison with the techniques of \cite{Ostrovsky-InsdelLDC-Compiler}.} The Insdel LDC construction of \cite{Ostrovsky-InsdelLDC-Compiler} also uses Schulman-Zuckerman (SZ) \cite{SchZuc99} codes, except it opens them up and directly uses the inefficient greedy inner codes used for the final efficient SZ codes themselves. In our case, we observe that the efficiently decodable codes of \cite{SchZuc99} have the additional property described in \Cref{lem:SZ-inner-code}, which states that small blocks have large weight. 
This observation implies a running time that is polynomial in the query complexity of the final codes, since it helps make the buffer-finding algorithms local.
The analysis of \cite{Ostrovsky-InsdelLDC-Compiler} also uses a binary search component, but our analysis and their analysis differ significantly.

\subsection{Related work} The study of codes for insertions and deletions was initiated by Levenstein \cite{Levenshtein_SPD66} in the mid 60's. Since then there has been a large body of works concerning insdel codes, and we refer the reader to the excellent surveys of \cite{Sloane2002OnSC,Mercier2010ASO,Mitzenmachen-survey}.
In particular, random codes with positive rate correcting from a large fraction of deletions were studied in \cite{Kiwi_expectedlength,guruswami2017deletion}. Efficiently encodable/decodable  codes, with constant rate, and that can withstand a constant fraction of insertion and deletions were extensively studied in \cite{SchZuc99,guruswami2017deletion,HaeuplerS17, ChengJLW18,HaeuplerS18, HaeuplerS18, ChengHLSW19,GuruswamiL19,BrakensiekGZ18,GuruswamiHS20,cheng2020efficient, ChengGHL20}. A recent area of interest is building ``list-decodable'' insdel codes, that can withstand a larger fraction of insertions and deletions, while outputting a small list of potential codewords  \cite{HaeuplerSS18,GuruswamiHS20,LiuTX20}.

In \cite{HaeuplerS18}, Haeupler and Shahrasbi construct explicit synchronization strings which can be locally decoded, in the sense that each index of the string can be computed using values located at only a small number of other indices. Synchronization strings are powerful combinatorial objects that can be used to index elements in constructions of insdel codes. These explicit and locally decodable synchronizations strings were then used to imply near linear time interactive coding scheme for insdel errors.

Recently, in \cite{ChengLZ20}, Cheng, Li and Zheng propose the notion of locally decodable codes with randomized encoding, in both the Hamming and edit distance regimes. They study such codes in the settings in which the encoder and decoder share randomness, or the channel is blivious to the codeword, and hence  adds error patterns non-adaptively. For edit error they obtain codes with $n=O(k)$ or $n= k \log k$ and $poly \log k$ query complexity.

There are various other notions of ``noisy search'' that have been studied in the literature.
Dhagat, Gacs, and Winkler~\cite{SODA:DhaGacWin92} consider a noisy version of the game ``Twenty Questions''.
In this problem, an algorithm searches an array for some element $x$, and a bounded number of incorrect answers can be given to the algorithm queries, and the goal is to minimize the number of queries made by an algorithm.
Feige \etal~\cite{FRPU94} study the depth of {\em noisy decision trees}: decision trees where each node gives the incorrect answer with some constant probability, and moreover each node success or failure is independent.
Karp and Kleinberg~\cite{SODA:KarKle07} study noisy binary search where direct comparison between elements is not possible; instead, each element has an associated biased coin.
Given $n$ coins with probabilities $p_1 \leq \dotsc \leq p_n$, target value $\tau \in [0,1]$, and error $\eps$, the goal is to design an algorithm which, with high probability, finds index $i$ such that the intervals $[p_i, p_{i+1}]$ and $[\tau-\eps, \tau+\eps]$ intersect.
Braverman and Mossel~\cite{SODA:BraMos08}, Klein \etal~\cite{ESA:KPSW11} and Geissmann \etal~\cite{GLLP17} study noisy sorting in the presence of recurrent random errors: when an element is first queried, it has some (independent) probability of returning the incorrect answer, and all subsequent queries to this element are fixed to this answer.
We note that each of the above notions of ``noisy search'' are different from each other and, in particular, different from our noisy search.


\subsection{Organization}
We begin with some general preliminaries in \cref{sec:prelims}. 
In \cref{sec:encoder-decoder} we present the formal encoder and decoder. In \cref{sec:block-decomposition} we define block decomposition which play an important role in our analysis. 
In \cref{sec:outer-decoder}, \cref{sec:noisy-binary-search}, and \cref{sec:buffer-detection} we prove correctness of our local decoding algorithm in a top-down fashion. 

\section{Preliminaries}\label{sec:prelims}
For integers $a \leq b$, we let $[a,b]$ denote the set $\{a, a+1, \dotsc, b\}$.
For positive integer $n$ we let $[n] \defeq [1,n]$.
All logarithms are base $2$ unless specified otherwise. 
We denote $x\circ y$ as the concatenation of string $x$ with string $y$. 
For any $x \in \Sigma^n$, $x[i] \in \Sigma$ denotes the $i$\th coordinate of $x$. 
Further, for $i < j$, we let $x[i,j] = (x_i, x_{i+1}, \dotsc, x_j)$ denote coordinates $i$ through $j$ of $x$.
A function $f(n)$ is said to be \emph{negligible} in $n$ if $f(n) = o\left(n^{-d}\right)$ for any $d \in \mathbb{N}$.
We let $\negl(n)$ denote an unspecified negligible function.
For any $x, y \in \Sigma^n$, $\HAM(x,y) = |\{i:x[i] \not = y[i]\}|$ denotes the Hamming distance between $x$ and $y$. 
Furthermore, $\ED(x,y)$ denotes the edit distance between $x$ and $y$; i.e., the minimum number of symbol insertions and deletions to transform $x$ into $y$. 
For any string $x \in \Sigma^*$ with finite length, we denote $|x|$ as the length of $x$. 
The fractional Hamming distance (resp., edit distance) is $\HAM(x,y)/|x|$ (resp., $\ED(x,y)/(2|x|)$).

\begin{definition}[Locally Decodable Codes for Hamming and Insdel errors]\label{def:ldc}
A code with encoding function $E: \Sigma_{M}^k\rightarrow \Sigma_{C}^n$ is a {\em $(q, \delta, \epsilon)$-Locally Decodable Code} (LDC) if there exists a randomized decoder $\mathcal D$, such that for every message $x\in \Sigma_{M}^k$ and index $i\in [k]$, and for every $w\in \Sigma_{C}^*$ such that $\dist(w, E(x))\leq \delta$ the decoder makes at most $q$ queries to $w$ and outputs $x_i$ with probability $\frac{1}{2}+\epsilon$; when $\dist$ is the fractional Hamming distance then this is a Hamming LDC; when $\dist$ is the fractional edit distance then this is an Insdel LDC. 
We also say that the code is binary if $\Sigma_{C}=\set{0,1}$.
\end{definition}

\begin{definition}[Locally Correctable Codes for Hamming and Insdel errors]\label{def:lcc}
A code with encoding function $E: \Sigma_{M}^k\rightarrow \Sigma_{C}^n$ is a {\em $(q, \delta, \epsilon)$-Locally Correctable Code} (LCC) if there exists a randomized decoder $\mathcal D$, such that for every message $x\in \Sigma_{M}^k$ and index $j\in [n]$,  and for every $w\in \Sigma_{C}^*$ such that $\dist(w, E(x))\leq \delta$ the decoder makes at most $q$ queries to $w$ and outputs $E(x)_j$ with probability $\frac{1}{2}+\epsilon$; when $\dist$ is the fractional Hamming distance then this is a Hamming error LCC; when $\dist$ is the fractional edit distance then this is an Insdel LCC. 
We also say that the code is binary if $\Sigma_{C}=\set{0,1}$.
\end{definition}

Our construction, like most insdel codes in the literature, is obtained via adaptations of the simple but powerful operation of code concatenation. 
If $C_{out}$ is an ``outer'' code over alphabet $\Sigma_{out}$ with encoding function $E_{out}: \Sigma_{out}^k\rightarrow \Sigma_{out}^n$, and $C_{in}$ is an ``inner'' code over alphabet $\Sigma_{in}$ with encoding function $E_{in}: \Sigma_{out}\rightarrow \Sigma_{in}^p$, then the concatenated code $C_{out}\bullet C_{in}$ is the code whose codewords lie in $\Sigma_{in}^{np}$, obtained by first applying $E_{out}$ to the message, and then applying $E_{in}$ to each symbol of the resulting outer codeword.

\section{Insdel LDCs/LCCs{} from Hamming LDCs/LCCs} \label{sec:encoder-decoder}
We give our main construction of Insdel LDCs/LCCs from Hamming LDCs/LCCs.
Our construction can be viewed as a procedure which, given outer codes $C_{out}$ and binary inner codes $C_{in}$ satisfying certain properties, produces binary codes $C(C_{out}, C_{in})$. 
This is formulated in the following theorem, which implies \Cref{thm:general}.

\begin{restatable}{theorem}{main}\label{thm:main}
Let $C_{out}$ and $C_{in}$ be codes such that
\begin{itemize}
    \item $C_{out}$ defined by $\Enc_{out} \colon \Sigma^k\rightarrow \Sigma^m$ is an a $(\ell_{out}, \delta_{out}, \epsilon_{out})$-LDC/LCC (for Hamming errors).
    \item $C_{in}$ is family of binary polynomial-time encodable/decodable codes with rate $1/\beta_{in}$ capable of correcting $\delta_{in}$ fraction of insdel errors. In addition, there are constants $\alpha_1, \alpha_2 \in (0,1)$ such that for any codeword $c$ of $C_{in}$, any substring of $c$ with length at least $\alpha_1|c|$ has fractional Hamming weight at least $\alpha_2$.
\end{itemize}
Then $C(C_{out}, C_{in})$ is a binary $\tp{\ell_{out} \cdot O\tp{\log^4 n'}, \Omega(\delta_{out}\delta_{in}), \epsilon-\negl(n')}$-Insdel LDC, or a binary $\tp{\ell_{out} \cdot O\tp{\log^5 n'}, \Omega(\delta_{out}\delta_{in}), \epsilon-\negl(n')}$-Insdel LCC, respectively. Here the codewords of $C$ have length $n = \beta m$ where $\beta = O\tp{\beta_{in}\log|\Sigma|}$, and $n'$ denotes the length of received word. 
\end{restatable}

For the inner code, we make use of the following efficient code constructed by  Schulman-Zuckerman \cite{SchZuc99}. 

\begin{lemma}[SZ-code~\cite{SchZuc99}]\label{lem:SZ-inner-code}
There exist constants $\beta_{in} \ge 1$, $\delta_{in} > 0$, such that for large enough values of $t > 0$, there exists a code $SZ(t) = (\Enc, \Dec)$ where $\Enc:\{0,1\}^{t} \rightarrow\{0,1\}^{\beta_{in}t}$ and $\Dec:\{0,1\}^{\beta_{in}t} \rightarrow\{0,1\}^{t}\cup \{\bot\}$ capable of correcting $\delta_{in}$ fraction of insdel errors, having the following properties:
\begin{enumerate}
	\item $\Enc$ and $\Dec$ run in  time $poly(t)$;
	\item For all $x \in \{0,1\}^t$, every interval of length $2\log t$ of $\Enc(x)$ has fractional Hamming weight at least $2/5$. 
\end{enumerate}
\end{lemma}
We formally complete the proof of correctness of \cref{thm:main} in \cref{sec:outer-decoder}. 
{\color{black}We only prove the correctness of the LDC decoder since it is cleaner and captures the general strategy of the LCC decoder as well. }
We dedicate the remainder of this section to outlining the construction of the encoding and decoding algorithms.

\subsection{Encoding and Decoding Algorithms}
In our construction of $C(C_{out}, C_{in})$, we denote 
the specific code of \Cref{lem:SZ-inner-code} as our inner code
$C_{in} = (\Enc_{in}, \Dec_{in})$. 
For our purpose, we view a message $x \in \Sigma^m$ as a pair in $[m] \times \Sigma^{\log m}$. 
The encoding function $\Enc_{in} \colon [m] \times \Sigma^{\log m} \rightarrow \set{0,1}^{\beta_{in}\tp{1 + \log|\Sigma|}\log m}$ maps a string in $\Sigma$ of length $\log m$ appended with an index from set $[m]$ ---  i.e., a (padded) message of bit-length $\tp{1 + \log|\Sigma|}\log m$ --- to a binary string of length $\beta_{in}\tp{1 + \log|\Sigma|}\log m$. The inner decoder $\Dec_{in}$ on input $y'$ returns $x$ if $\ED\tp{y', y} \le \delta_{in}\cdot 2|y|$ where $y = \Enc_{in}(x)$. 
The information rate of this code is $R_{in} = 1/\beta_{in}$.\mbox{}\\

\noindent\textbf{\textsf{The Encoder $(\Enc)$.}} 
Given an input string $x \in \Sigma^k$ and outer code $C_{out} = (\Enc_{out}, \Dec_{out})$, our final encoder \Enc does the following: 
\begin{enumerate}
	\item Computes the outer encoding of $x$ as $s = \Enc_{out}(x)$;
	
	\item For each $i \in [m/\log m]$, groups $\log m$ symbols $s[(i-1)\log m,i\log m -1]$ into a single block $b_i \in \Sigma^{\log m}$;
	
	\item For each $i \in [m/\log m]$, computes the $i^{th}$ block of the inner encoding as $Y\itn{i} = \Enc_{in}(i \circ b_i)$ --- i.e., computes the inner encoding of the $i$\th block concatenated with the index $i$;
	
	\item For some constant $\alpha \in (0,1)$ (to be decided), appends a $\alpha\log m$-long buffer of zeros before and after each block; and
	
	\item Outputs the concatenation of the buffered blocks (in indexed order) as the final codeword $c = \Enc(x) \in \{0,1\}^{n}$, where
	\begin{align}
		& &c = {{\color{gray}\bigg(}0^{\alpha\log m} \circ Y\itn{1} \circ 0^{\alpha\log m}{\color{gray}\bigg)}\circ \cdots \circ {\color{gray}\bigg(}0^{\alpha\log m} \circ Y\itn{m/\log m} \circ 0^{\alpha\log m}{\color{gray}\bigg)}.} \label{eq:codeword}
	\end{align}
\end{enumerate}
Denoting $\beta = 2\alpha + \beta_{in}\tp{1 + \log|\Sigma|}$, the length of $c = \Enc(x)$ is 
$$ n = \tp{2\alpha\log m + \beta_{in}\tp{1 + \log|\Sigma|}\log m} \cdot \frac{m}{\log m} = \beta m. $$

\noindent\textbf{\textsf{The LDC Decoder $(\Dec)$.}}
We start off by describing the high-level overview of our decoder $\Dec$ and discuss the challenges and solutions behind its design. 
As defined in \cref{eq:codeword}, our encoder \Enc, on input $x \in \Sigma^k$, outputs a codeword $c = c_1 \circ \cdots \circ c_{d} \in \zo^n$, where $d = m/\log m$.
The decoder setting is as follows:
on input $i \in [k]$ and query access to the corrupted codeword $c' \in \zo^{n'}$ such that $\ED(c,c') \leq 2n\delta$, our final decoder \Dec{} needs to output the message symbol $x[i]$ with high probability.
Notice that if \Dec{} had access to the original codeword $s = \Enc_{out}(x)$, then \Dec{} could simply run $\Dec_{out}(i)$ while supplying it with oracle access to this codeword $s$.
This naturally motivates the following decoding strategy: simulate oracle access to the codeword $s$ by answering the queries of $\Dec_{out}$ by decoding the appropriate bits using $\Dec_{in}$. 
We give a detailed description of this strategy next.

Let $Q_i = \{q_1,\dotsc, q_{\ell_{out}}\} \subset[m]$ be a set of indices which $\Dec_{out}(i)$ queries.\footnote{Our construction also supports adaptive queries, but we use non-adaptive queries for ease of presentation.}
We observe that if our decoder had oracle access to the uncorrupted codeword $c$, then answering these queries would be simple:
\begin{enumerate}
	\item For each $q \in Q_i$, let $b_j = s[(j-1)\log m, j\log m-1]$ be the block which contains $s[q]$.
	In particular, $q = (j-1)\log m + r_j$ for some $r_j \in [0, \log m-1]$,
	
	\item Obtain block $c_j$ by querying oracle $c$ and obtain $Y\p{j}$ by removing the buffers from $c_j$,
	
	\item Obtain $j \circ b_j$ by running $\Dec_{in}(Y\p{j})$, then return $s[q] = b_j[r_j]$ to $\Dec_{out}$.
\end{enumerate}

In fact, it suffices to answer the queries of $\Dec_{out}$ with symbols consistent with any string $s'$ such that $\HAM(s,s') \leq m\delta_{out}$.
Then the correctness of the output would follow from the correctness of $\Dec_{out}$. 
We carry out the strategy mentioned above, except that now we are given a corrupted codeword $c'$.

For the purposes of analysis, we first define the notion of a {\em block decomposition} of the corrupted codeword $c'$.
Informally, a block decomposition is simply a partitioning of $c'$ into contiguous blocks.
Our first requirement for successful decomposition is that there must exist a block decomposition $c' = c'_1 \circ \cdots \circ c'_{d}$ that is ``not too different'' from the original decomposition $c = c_1 \circ \cdots \circ c_d$.\footnote{We note that we do not need to know this decomposition explicitly, and that its existence is sufficient for our analysis.}
In particular, we require that $\sum_j \ED(c'_j, c_j) \leq 2n\delta$, which is guaranteed by \cref{prop:block-decomps-exist}.
Next, we define the notion of {\em $\gamma$-good} (see \cref{def:g-good}).
The idea here is that if a block $c'_j$ is $\gamma$-good (for appropriate $\gamma$), then we can run $\Dec_{in}$ on $c'_j$ and obtain $j \circ b_j$. 
As the total number of errors is bounded, it is easy to see that all but a small fraction of blocks are $\gamma$-good (\cref{lem:g-bad}). 
At this point, we are essentially done if we can decode $c_j'$ for any given $\gamma$-good block $j$.

An immediate challenge we are facing is that of {\em locating} a specific $\gamma$-good block $c'_j$, while maintaining overall locality. The presence of insertions and deletions may result in uneven block lengths and misplaced blocks, making the task of locating a specific block non-trivial.
However, $\gamma$-good blocks make up the majority of the blocks and enjoy the property that they are in correct relative order, it is conceivable to perform a \emph{binary search} style of algorithm over the blocks of $c'$ to find block $c'_j$. 
The idea is to maintain a search interval and iteratively reduce its size by a constant multiplicative factor. 
In each iteration, the algorithm samples a small number of blocks and obtains their (appended) indices. 
As the vast majority of blocks are $\gamma$-good, these indices guide the binary search algorithm in narrowing down the search interval. 
Though there is one problem with this argument: the density of $\gamma$-good blocks may decrease as the search interval becomes smaller. 
In fact, it is impossible to \emph{locally} locate a block $c'_j$ surrounded by many bad blocks, even if $c'_j$ is $\gamma$-good. 
This is where the notion of $(\theta, \gamma)$-locally good (see \cref{def:locally-good-block}) helps us: if a block $c'_j$ is $(\theta,\gamma)$-locally good, then $(1-\theta)$-fraction of blocks in every neighborhood around $c'_j$ are $\gamma$-good, and every neighborhood around $c'_j$ has a bounded number of errors.
Therefore, as long as the search interval contains a locally good block, we can lower bound the density of $\gamma$-good blocks and recover $c'_j$ with high probability.

Our {\em noisy binary search algorithm} essentially implements this idea.
On input block index $j$, the algorithm searches for block $j$.
If block $j$ is $(\theta,\gamma)$-locally good, then we can guarantee that our noisy binary search algorithm will find $j$ except with negligible probability (see \cref{thm:main-noisy-binary-search}).
Thus it is desirable that the number of $(\theta,\gamma)$-locally good blocks is large; if this number is large, the noisy binary search is effectively providing oracle access to a string $s'$ which is close to $s$ in Hamming distance, and thus the outer decoder is able to decode $x[i]$ with high probability.
\cref{lem:bad-local} exactly guarantees this property.

The discussion above requires knowing the boundaries of each block $c'_j$, which is non-trivial even in the no corruption case.
As the decoder is oblivious to the block decomposition, the decoder works with {\em approximate boundaries} which can be found locally by a {\em buffer search algorithm}, described as follows.
Recall that by construction $c_j$ consists of $Y\p{j}$ surrounded by buffers of $(\alpha \log m)$-length 0-runs.
So to find $Y\p{j}$, it suffices to find the buffers surrounding $Y\p{j}$.
Our buffer search algorithm can be viewed as a ``local variant'' of the buffer search algorithm of Schulman and Zuckerman \cite{SchZuc99}. 
This algorithm is designed to find approximate buffers surrounding a block $c'_j$ if it is $\gamma$-good. 
Then the string in between two buffers is identified as a corrupted codeword and is decoded to $j \circ b_j$.
The success of the algorithm depends on $\gamma$-goodness of the block being searched and requires that any substring of a codeword from $C_{in}$ has ``large enough'' Hamming weight.
In fact, our inner code given by \cref{lem:SZ-inner-code} gives us this exact guarantee.
All together, this enables the noisy binary search algorithm to use the buffer finding algorithm to search for a block $c'_j$.

We formalize the decoder outlined above.
On input $i \in [k]$, $\Dec$ simulates $\Dec_{out}(i)$ and answers its queries. 
Whenever $\Dec_{out}(i)$ queries an index $j \in [m]$, $\Dec$ expresses $j = (p-1)\log m + r_j$ for $p \in [m/\log m]$ and $r_j \in [0,\log m - 1]$, and runs $\noisyBinarySearch(c', p)$ (which calls the algorithm \buffind) to obtain a string $b' \in \Sigma^{\log m}$ (or $\perp$). 
Then it feeds the $r_j$-th symbol of $b'$ (or $\perp$) to $\Dec_{out}(i)$. 
Finally, $\Dec$ returns the output of $\Dec_{out}(i)$. \mbox{}\\

\noindent\textbf{\textsf{The LCC Decoder $(\Dec)$.}}
Similar to the LDC decoder, our LCC decoder \Dec{} does the following: let $B = 2\alpha\log m + \beta_{in}\tp{1 + \log|\Sigma|}\log m$. 
On input $j \in [n]$, $\Dec$ first expresses $j = (p-1)B + r_j$ for some $p \in [m/\log m]$ and $0\le r_j < B$, and checks whether $j$ is inside a buffer. 
Specifically, if $r_j \in [0, \log m) \cup [B-\log m, B)$ then it outputs 0. 
Otherwise, it simulates $\Dec_{out}((p-1)\log m + r)$ for each $0\le r< \log m$, and answers their queries. 
Whenever $\Dec_{out}$ queries $i \in [m]$, $\Dec$ expresses $i = (b-1)\log m + r_i$ for some $b \in [m/\log m]$ and $0\le r_i < \log m$, and runs $\noisyBinarySearch(c', b)$ to obtain a string $S \in \Sigma^{\log m}$ (or $\perp$), and answers the query with $S_{r_i}$ (or $\perp$). 
Finally, denoting by $s_r$ the output of $\Dec_{out}((p-1)\log m + r)$, $\Dec$ returns the $(r_j-\log m+1)$-th bit of $\Enc_{in}\tp{p \circ s_0s_1\ldots s_{\log m - 1}}$. \mbox{}\\

\noindent\textbf{\textsf{Efficiency.}}
We note that the efficiency of our compiler depends on the efficiency of the inner and outer codes. 
Let $T(\Enc_{in}, l)$, $T(\Enc_{out},l)$, $T(\Enc, l)$ denote the run-times of the inner, outer and final encoders, respectively, on inputs of length $l$. 
Similarly, let $T(\Dec_{in}, l)$, $T(\Dec_{out},l)$, $T(\Dec, l)$ denote the run-times of the inner, outer, and final decoders (respectively), with oracle access to corrupted codewords of length $l$. 
Then we have following run-time relations:
\begin{align*}
    & &T(\Enc, k) & = T(\Enc_{out}, k) + O(m/\log m) \cdot T(\Enc_{in}, \log|\Sigma| \cdot \log m + \log m),& \\
    & &T(\Dec, n') &= T(\Dec_{out},m) + \ell_{out} \cdot O\tp{\log^3n'} \cdot T(\Dec_{in}, \beta\log m).&
\end{align*}
Here, $n'$ is the length of the corrupted codeword, $k$ is the input length of $\Enc_{out}$, $m$ is the input length of $\Enc_{in}$, $\ell_{out}$ is the locality of $\Dec_{out}$, and $1/\beta$ is the rate of the final encoder.

\section{Block Decomposition of Corrupted Codewords} \label{sec:block-decomposition}
The analysis of our decoding procedure relies on a so-called buffer finding algorithm and a noisy binary search algorithm.
To analyze these algorithms, we introduce the notion of a {\em block decomposition} for (corrupted) codewords, as well as what it means for a block to be {\em (locally) good}.

For convenience, we now fix some notation for the remainder of the paper. 
We fix an arbitrary message $x \in \Sigma^k$. 
We use $s = \Enc_{out}(x) \in \Sigma^{m}$ for the encoding of $x$ by the outer encoder. 
Let $\tau = \log m$ be the length of each block and $d = m/\log m$ be the number of blocks. 
For $i \in [d]$, we let $b_i \in \Sigma^{\tau}$ denote the $i$-th block $s[(i-1)\tau,i\tau-1]$, and let $Y\itn{i}$ denote the encoding $\Enc_{in}\tp{i \circ b_i}$. 
Recall that $\alpha\tau$ is the length of the appended buffers for some $\alpha \in (0,1)$, and the parameter $\beta = 2\alpha + \beta_{in}(1 + \log|\Sigma|)$. 
Thus $|Y\itn{i}| = (\beta-2\alpha)\tau$. 
The final encoding is given by
\begin{align*}
& &c = \tilde{Y}\itn{1} \circ \tilde{Y}\itn{2} \circ \cdots \circ \tilde{Y}\itn{d},
\end{align*}
where $\tilde{Y}\itn{j} = 0^{\alpha\tau} \circ Y\itn{j} \circ 0^{\alpha\tau}$ and $|\tilde{Y}\p{j}| = \beta\tau$. The length of $c$ is $n=d\beta\tau = \beta m$. 
We let $c' \in \set{0,1}^{n'}$ denote a corrupted codeword satisfying $\ED\tp{c,c'} \le 2n\cdot \delta$.

\begin{definition}[Block Decomposition]\label{def:block-decomp}
    A {\em block decomposition} of a (corrupted) codeword $c'$ is a non-decreasing mapping $\phi \colon [n'] \rightarrow [d]$ for $n', d \in \bbZ^+$. 
\end{definition}

We say a set $I \subseteq [n']$ is an interval if $I = \emptyset$ (i.e., an empty interval) or $I = \set{l, l+1, \ldots, r-1}$ for some $1 \le l < r \le n'$, in which case we write $I=[l,r)$.
For an interval $I=[l,r)$, we write $c'[I]$ for the substring $c'[l] c'[l+1]\ldots c'[r-1]$.
Finally, $c[\emptyset]$ stands for the empty string. 

We remark that for a given block decomposition $\phi$, since $\phi$ is non-decreasing we have that for every $j \in [d]$ the pre-image $\phi^{-1}(j)$ is an interval. 
Since $\phi$ is a total function, it induces a partition of $[n']$ into $d$ intervals $\set{\phi^{-1}(j) \colon j \in [d]}$. 
The following definition plays an important role in the analysis.
\begin{definition}[Closure Intervals] 
	The {\em closure of an interval $I=[l,r) \subseteq [n']$} is defined as $\cup_{i=l}^{r-1}\phi^{-1}(\phi(i)).$ 
	An interval $I$ is a \emph{closure interval} if the closure of $I$ is itself. 
	Equivalently, every closure interval has the form $\+I[a,b] \coloneqq \bigcup_{j=a}^{b}\phi^{-1}(j)$ for some $a, b \in [d]$.
\end{definition}

\begin{proposition}\label{prop:block-decomps-exist}
    There exists a block decomposition $\phi \colon [n'] \rightarrow [d]$ such that 
    \begin{align*}
    	& &\sum_{j \in [d]}\ED\tp{c'[\phi^{-1}(j)], \ \tilde{Y}\itn{j}} \le \delta \cdot 2n.
    \end{align*}
\end{proposition}
\begin{proof}
	Let $\phi_0 \colon [n] \rightarrow [d]$ be the block decomposition for $c$ satisfying $\phi_0(i) = j$ if $i$ lies in block $\tilde{Y}\p{j}$.
	Without loss of generality, we assume the adversary performs the following corruption process: 
	\begin{enumerate}
	    \item The adversary picks some $j \in [d]$;
	    \item The adversary corrupts $\tilde{Y}\p{j}$.
	\end{enumerate}
	Steps (1) and (2) are repeated up to the specified edit distance bound of $2\delta n$.
	We construct $\phi \colon [n'] \rightarrow [d]$ by modifying the decomposition $\phi_0$ according to the above process.
	It is clear that $\phi$ satisfies the desired property.
\end{proof}

We now introduce the notion of {\em good blocks}.
In the following definitions, we also fix an arbitrary block decomposition $\phi$ of $c'$ enjoying the property guaranteed by \Cref{prop:block-decomps-exist}.

\begin{definition}[$\gamma$-good block]\label{def:g-good}
	For $\gamma \in (0,1)$ and $j \in [d]$ we say that block $j$ is {\em $\gamma$-good} if $\ED( c'[\phi^{-1}(j)], \tY\p{j} ) \leq \gamma\alpha\tau$.
	Otherwise we say that block $j$ is {\em $\gamma$-bad}.
\end{definition}

\begin{definition}[$(\theta, \gamma)$-good interval]\label{def:good-interval}
    We say a closure interval $\+I[a,b]$ is $(\theta, \gamma)$-good if the following hold:
    \begin{enumerate}
        \item $\sum_{j=a}^{b}\ED\tp{c'[\phi^{-1}(j)], \tilde{Y}\itn{j}} \le \gamma \cdot (b-a+1)\alpha\tau $.
        \item There are at least $(1-\theta)$-fraction of $\gamma$-good blocks among those indexed by  $\set{a, a+1, \cdots, b}$. 
    \end{enumerate}
\end{definition}

\begin{definition}[$(\theta,\gamma)$-local good block] \label{def:locally-good-block}
    For $\theta, \gamma \in (0,1)$ we say that block $j$ is {\em $(\theta,\gamma)$-local good} if for every $a,b \in [d]$ such that $a \leq j \leq b$ the interval $\+I[a,b]$ is $(\theta,\gamma)$-good.
    Otherwise, block $j$ is $(\theta,\gamma)$-locally bad.
\end{definition}
Note that in \cref{def:locally-good-block}, if $j$ is $(\theta,\gamma)$-locally good, then $j$ is also $\gamma$-good by taking $a = b = j$.

\begin{proposition}\label{prop:bounds-good}
The following bounds hold:
\begin{enumerate}
\item For any $\gamma$-good block $j$, $(\beta-\alpha\gamma)\tau \leq |\phi^{-1}(j)| \leq (\beta+\alpha\gamma)\tau$.
\item For any $(\theta, \gamma)$-good interval $\+I[a,b]$, $(b-a+1)(\beta-\alpha\gamma)\tau \le \abs{\+I[a,b]} \le (b-a+1)(\beta+\alpha\gamma)\tau$.
\end{enumerate}
\end{proposition}
\begin{proof}
    For item (1) note that an uncorrupted block has length $\beta\tau$. Since $j$ is $\gamma$-good, we know that $\ED(c'[\phi^{-1}(j)], \tY\itn{j}) \leq \gamma \alpha\tau$, which implies that $(\beta-\alpha\gamma)\tau \leq |\phi^{-1}(j)| \leq (\beta+\alpha\gamma)\tau$.
    
    For item (2), we first note that $|a| - \Delta \le |b| \le |a| + \Delta$ where $\Delta = \ED\tp{a, b}$.
    Let $\Delta_j = \ED\tp{c'[\phi^{-1}(j)], \tY\itn{j}}$. By definition of $(\theta,\gamma)$-good interval, we have that $\sum_{j=a}^b \Delta_j \leq \gamma (b-a+1)\alpha\tau$. This gives us the following two properties. 
    \begin{align*}
        & &\abs{\+I[a,b]} = \sum_{j=a}^{b}\abs{\phi^{-1}(j)} \le \sum_{j=a}^{b}\beta\tau + \Delta_j \le (b-a+1)(\beta+\alpha\gamma)\tau,\\
        & &\abs{\+I[a,b]} = \sum_{j=a}^{b}\abs{\phi^{-1}(j)} \ge \sum_{j=a}^{b}\beta\tau - \Delta_j \ge (b-a+1)(\beta-\alpha\gamma)\tau.
    \end{align*}
\end{proof}

The following lemmas give upper bounds on the number of $\gamma$-bad and $(\theta,\gamma)$-locally bad blocks.
\begin{lemma}\label{lem:g-bad}
	The total fraction of $\gamma$-bad blocks is at most $2\beta\delta/(\gamma\alpha)$.
\end{lemma}
\begin{proof}
	Let $\Delta_j = \ED( c'[\phi^{-1}(j)], \tY_j )$ for every $j \in [d]$.
	By our choice of $\phi$ and \Cref{prop:block-decomps-exist} we have that:
 	\begin{align*}
 		& &\sum_{j=1}^{d} \Delta_j \leq  2n\cdot \delta.
 	\end{align*}
	Let $\mathsf{Bad} \subseteq [d]$ be the set of $\gamma$-bad blocks.
	Then we have
 	\begin{align*}
		& &\delta\cdot 2n \ge \sum_{j=1}^d \Delta_j \geq \sum_{i \in \mathsf{Bad}} \Delta_i > \abs{\mathsf{Bad}} \cdot \gamma \alpha\tau
 	\end{align*}
	where the latter inequality follows by the definition of $\gamma$-bad.
	Thus we obtain $\abs{\mathsf{Bad}} < \delta n / \gamma\alpha \tau$.
	Recalling that $n = \beta d\tau$ 
	we have that $\abs{\mathsf{Bad}}/d < 2\beta\delta/(\gamma\alpha)$ as desired.
\end{proof}

\begin{lemma}\label{lem:bad-local}
    The total fraction of $(\theta,\gamma)$-local bad blocks is at most $(4/\gamma\alpha)(1+1/\theta)\delta\beta$.
\end{lemma}
\begin{proof}
	First we count the number of blocks which violate condition (1) of \Cref{def:good-interval}.
	We proceed by counting in two steps.
	Suppose that $i_1 \in [d]$ is the smallest index such that block $i_1$ violates (1) of \Cref{def:good-interval} with witness $(i_1,b_1)$; that is, $\ED_{j=i_1}^{b_1} \ED(c'[\phi^{-1}(j), \tY\p{j}]) > \gamma \cdot (b_1-i_1+1)\tau$.
	Continuing inductively, let $i_k \in [d]$ be the smallest index such that $i_k > i_{k-1}+b_{k-1}$ and $i_k$ violates condition (1) of \Cref{def:good-interval} with witness $(i_k, b_k)$.
	Let $\set{(i_k, b_k)}_{k=1}^t$ for some $t$ be the result of this procedure.
	Further let $D_k = \sum_{i=i_k}^{b_k} \ED\tp{c'[\phi^{-1}(i), \tY\p{i}}$ for every $k \in [t]$.
	Let $n_\gamma\p{1}$ be the total number of locally bad blocks $j$ of the form $(j,b)$ for some $b$.
	Then we claim that (1) $n\p{1}_\gamma \leq \sum_{k=1}^t b_k - i_k$, (2) for all $k \in [t]$ we have that $D_k > \gamma\tau (b_k-i_k)$, and (3) $\sum_{k=1}^t D_k \leq \ED(c,c')$.
	The first equation follows from the fact that any locally bad block $j$ with witnes $(j,b)$ for some $b \geq j$ must fall into some interval $[i_k, b_k]$, else this would contradict the minimality of the chosen $i_k$.
	The second equation follows directly by definition of local good.
	The third equation follows from the fact that the sum of $D_k$ is at most the sum of all possible blocks, which is upper bounded by the edit distance.
	Combining these equations we see that $n_\gamma\p{1} \leq 2\delta n/(\gamma\alpha\tau).$
	Symmetrically, we can consider all bad blocks $j$ which violate condition (1) of \Cref{def:good-interval} and have witnesses of the form $(a,j)$.
	For this bound we obtain $n_\gamma\p{2} \leq 2\delta n/(\gamma\alpha\tau).$
	
	Now we consider the number of bad blocks which violate condition (2) of \Cref{def:good-interval}.
	By identical analysis and first considering bad blocks $j$ with witnesses of the form $(j,b)$, we obtain a set of minimally chosen witnesses $\set{(i_k,b_k)}_{k=1}^t$.
	Let $n_\theta\p{1}$ be the total number of bad blocks $j$ with witnesses of the form $(j,b)$.
	Further, let $B_k$ denote the number of $\gamma$-bad blocks in the interval $[i_k, b_k]$.
	Then we have (1) $n\p{1}_\theta \leq \sum_{k=1}^t b_k-i_k$, (2) for all $k \in [t]$, $B_k > \theta(b_k - i_k)$, and (3) $\sum_{k=1}^t B_k\leq \ED(c,c')/(\gamma\alpha\tau)$.
	Then by these three equations we have that $n_\theta\p{1} \leq 2\delta n/(\gamma\theta\alpha\tau)$.
	By a symmetric argument, if $n_\theta\p{2}$ is the total number of blocks $j$ which violate condition (2) of \Cref{def:good-interval} with witnesses of the form $(a,j)$ then we have $n_\theta\p{2} \leq 2\delta n/(\theta\gamma\alpha\tau)$.
	
	Thus the total number of possible bad blocks violating either condition is at most $(4/\gamma\alpha\tau)(1 + 1/\theta)\delta n$.
	Recalling that $n = \beta d \tau$, we have that the total fraction of locally bad blocks is at most $(4/\gamma\alpha)(1+1/\theta)\delta\beta$ as desired.
\end{proof}

\section{Outer Decoder} \label{sec:outer-decoder}
At a high level, the our decoding algorithm $\Dec$ runs the outer decoder $\Dec_{out}$ and must answer all oracle queries of $\Dec_{out}$ by simulating oracle access to some corrupted string $s'$. 
Recall that $C_{out}$, with encoding function $\Enc_{out} \colon \Sigma^k \rightarrow \Sigma^m$, is a $(\ell_{out}, \delta_{out}, \eps_{out})$-LDC for Hamming errors. 
Further, $C_{out}$ has probabilistic decoder $\Dec_{out}$ such that for any $i \in [k]$ and string $s' \in \tp{\Sigma \cup \set{\perp}}^{m}$ such that $\textsf{HAM}\tp{s', s} \le m\cdot \delta_{out}$ for some codeword $s = \Enc_{out}(x)$, we have
\begin{align*}
	& &\Pr\left[ \Dec_{out}^{s'}(i) = x[i] \right] \ge \frac{1}{2} + \eps_{out}.
\end{align*}
Additionally, $\Dec_{out}$ makes at most $\ell_{out}$ queries to $s'$.

In order to run $\Dec_{out}$, we need to simulate oracle access to such a string $s'$. 
To do so, we present our noisy binary search algorithm \cref{alg:noisy-binary-search} in \cref{sec:noisy-binary-search}.
For now, we assume \cref{alg:noisy-binary-search} has the properties stated in the following proposition and theorem.
\begin{restatable}{proposition}{nbsQuery}\label{prop:noisybs-query-complexity}
	\Cref{alg:noisy-binary-search} has query complexity $O\tp{\log^4 n'}$.
\end{restatable}

\begin{restatable}{theorem}{nbsToOuter}\label{thm:noisybs-to-outer-decoder}
For $j \in [d]$, let $\*b_j \in \Sigma^{\tau} \cup \set{\perp}$ be the random variable denoting the output of \cref{alg:noisy-binary-search} on input $(c',1,n'+1,j)$. 
We have
\begin{align*}
& &\Pr\left[ \Pr_{j \in [d]}\left[\*b_j \neq b_j \right] \ge \delta_{out} \right] \le \negl(n'),
\end{align*}
where the probability is taken over the joint distribution of $\set{\*b_j \colon j \in [d]}$.
\end{restatable}
We note that in \cref{thm:noisybs-to-outer-decoder}, the random variables $\*b_j$ do not need to be independent, i.e., two runs of \cref{alg:noisy-binary-search} can be correlated. 
For example, we can fix the random coin tosses of \cref{alg:noisy-binary-search} before the first run and reuse them in each call.

\section{Noisy Binary Search} \label{sec:noisy-binary-search}
We present \cref{alg:noisy-binary-search} in this section. 
As mentioned in \cref{sec:outer-decoder}, the binary search algorithm discussed in this section can be viewed as providing the outer decoder with oracle access to some string $s' \in \tp{\Sigma \cup \set{\perp}}^{m}$. 
Namely whenever the outer decoder queries an index $j \in [m]$ which lies in block $p$, we run \textsc{Noisy-Binary-Search} on $(c',1,n'+1,p)$ and obtain a string $b'_p \in \Sigma^{\log m}$ which contains the desired symbol $s'[j]$. 

\begin{algorithm}[ht] 
	\caption{\sc{Noisy binary search}}
	\label{alg:noisy-binary-search}
	\textbf{Input:} An index $j \in [d]$, and oracle access to a codeword $c' \in \set{0, 1}^{n'}$. \\
	\textbf{Output:} A string $b \in \Sigma^{\tau}$ or $\perp$.
	\begin{algorithmic}[1]
		\State $N \gets \Theta(\log^2 n')$
		\State $\rho \gets \min\set{\frac{1}{4} \cdot \frac{\beta-\gamma}{\beta+\gamma}, 1-\frac{3}{4}\cdot\frac{\beta+\gamma}{\beta-\gamma}}$
		\State $C \gets 36(\beta+\gamma)\tau$
		\Function{Noisy-Binary-Search}{$c'$, $l$, $r$, $j$} 
		\If{$r - l \le C$}
			\State $s \gets$ \Call{Interval-Decode}{$l$, $r$, $j$} 
			\State \Return{$s$}
		\EndIf
		\State $m_1 \gets (1-\rho)l + \rho r$, $m_2 \gets \rho l + (1-\rho)r$
		\For{$t \gets 1$ to $N$}
		    \State Randomly sample $i$ from $\set{m_1, m_1+1, \ldots, m_2-1}$
		    \State $j_t \gets $ \Call{Block-Decode}{$i$}
		\EndFor
		\State $\tilde{j} \gets $ median of $j_1, \ldots, j_N$ (ignore $j_t$ if $j_t = \perp$)
		\If {$j \le \tilde{j}$}
			\State \Return \Call{Noisy-Binary-Search}{$c'$, $l$, $m_2$, $j$}
		\Else
			\State \Return \Call{Noisy-Binary-Search}{$c'$, $m_1$, $r$, $j$}
		\EndIf
		\EndFunction
	\end{algorithmic}
\end{algorithm}

We analyze the query complexity of \cref{alg:noisy-binary-search} and prove \cref{prop:noisybs-query-complexity}.
\nbsQuery*
\begin{proof}
	The number of iterations $T$ is at most $O\tp{\log \frac{n'}{C}} = O\tp{\log n'}$ as $r-l$ is reduced by a constant factor $1-\rho$ in each iteration until it goes below $C$. 
	In each iteration (except for the last iteration), the algorithm makes $N=\Theta(\log^2 n')$ calls to \textsc{Block-Decode}, which has query complexity $O\tp{\log n'}$. 
	In the last iteration, it calls \textsc{Interval-Decode} which has query complexity $O\tp{\log n'}$. 
	Thus the overall query complexity is $O\tp{\log^4 n'}$.
\end{proof}

The following theorem shows that the set of indices which can be correctly returned by \cref{alg:noisy-binary-search} is captured by the locally good property.

\begin{restatable}{theorem}{mainNoisyBinarySearch}\label{thm:main-noisy-binary-search}
	If $j \in [d]$ is a $(\theta,\gamma)$-locally good block, running \cref{alg:noisy-binary-search} on input $(c', 1, n'+1, j)$ outputs $b_j$ with probability at least $1-\negl(n')$. 
\end{restatable}
We defer the proof of \cref{thm:main-noisy-binary-search} to \cref{app:main-noisybs-proof}, as the proof requires many auxiliary claims and lemmas. 
For now, we assume \cref{thm:main-noisy-binary-search} and work towards proving \cref{thm:noisybs-to-outer-decoder}.

We first observe that the only time \cref{alg:noisy-binary-search} interacts with $c'$ is when it queries \textsc{Block-Decode} and \textsc{Interval-Decode}.
Thus the properties of these two algorithms is essential to our proof. 
We briefly describe these two subroutines now.
\begin{itemize}
	\item \textsc{Block-Decode}: On input index $i \in [n']$, \textsc{Block-Decode} tries to find the block $j$ that contains $i$, and attempts to decode the block to $j \circ b_j$. 
	It returns the index $j$ if the decoding was successful, and $\perp$ otherwise.
	
	\item \textsc{Interval-Decode}: On input $l, r \in [n']$ and $j \in [d]$, \textsc{Interval-Decode} (roughly) runs the buffer search algorithm of Schulman and Zuckerman \cite{SchZuc99} over the substring $c'[l,r]$ to obtain a set of approximate buffers, and attempts to decode all strings separated by the approximate buffers. 
	It returns $b$ if any string is decoded to $j \circ b$, and $\perp$ otherwise.
\end{itemize}

For convenience, we model \textsc{Block-Decode} as a function $\varphi \colon [n'] \rightarrow [d] \cup \set{\perp}$, and model \textsc{Interval-Decode} as a function $\psi \colon [n'] \rightarrow \Sigma^{\tau} \cup \set{\perp}$. 
The functions $\varphi$ and $\psi$ have the following properties, which are crucial to the proof of \cref{thm:noisybs-to-outer-decoder}.
\begin{restatable}{theorem}{buffProperties}\label{thm:buffer-finding-properties}
	The functions $\varphi$ and $\psi$ satisfy the following properties:
	\begin{enumerate}
		\item For any $\gamma$-good block $j$ we have 
		\begin{align*}
			& &\Pr_{i \in \phi^{-1}(j)}\left[ \varphi(i) \neq j \right] \le \gamma.
		\end{align*}
	
		\item Let $[l,r)$ be an interval with closure $\+I[L,R-1]$, satisfying that every block $j \in \set{L,\ldots, R-1}$ is $\gamma$-good. Then for every block $j$ such that $\phi^{-1}(j) \subseteq [l,r)$, we have $\psi(j,l,r)=b_j$.
	\end{enumerate}
\end{restatable}

Given \cref{thm:main-noisy-binary-search} and \cref{thm:buffer-finding-properties}, we recall and prove \cref{thm:noisybs-to-outer-decoder}.

\nbsToOuter*
\begin{proof}
	Let $\Goody \subseteq [d]$ be the set of $(\theta, \gamma)$-locally-good blocks, and let $\overline{\Goody} = [d] \setminus \Goody$. \Cref{lem:bad-local} implies that 
	\begin{align*}
		& &\abs{\overline{\Goody}} \le \tp{1 + \frac{1}{\theta}}\frac{\delta d\beta}{\alpha\gamma} = \frac{\delta_{out}d}{2}. 
	\end{align*}
	For each $j \in \Goody$, denote by $E_j$ the event $\set{\*b_j \neq b_j}$. \Cref{thm:main-noisy-binary-search} in conjunction with a union bound implies that
	\begin{align*}
		& &\Pr\left[ \bigcup_{j \in \Goody}E_j \right] \le \negl(n').
	\end{align*}
	Since
	\begin{align*}
		& &\Pr_{j \in [d]}\left[\*b_j \neq b_j \right] &\le \Pr_{j \in [d]}\left[ j \in \overline{\Goody} \right] + \Pr_{j \in [d]}\left[\*b_j \neq b_j \ \middle| \ j \in \Goody \right] \\
		& & &\le \frac{\delta_{out}}{2} + \Pr_{j \in [d]}\left[\*b_j \neq b_j \ \middle| \ j \in \Goody \right],
	\end{align*}
	we have
	\begin{align*}
		& &\Pr\left[ \Pr_{j \in [d]}\left[\*b_j \neq b_j \right] \ge \delta_{out} \right] &\le \Pr\left[ \Pr_{j \in [d]}\left[\*b_j \neq b_j \ \middle| \ j \in \Goody \right] \ge \frac{\delta_{out}}{2} \right] \\
		& & &\le \Pr\left[ \bigcup_{j \in \Goody}E_j \right] \le \negl(n').
	\end{align*}
\end{proof}

\section{Block Decode Algorithm} \label{sec:buffer-detection}
A key component of the Noisy Binary Search algorithm is the ability to decode $\gamma$-good blocks in the corrupted codeword $c'$.
In order to do so, our algorithm will take explicit advantage of the $\gamma$-good properties of a block.
We present our block decoding algorithm, named \bbs, in \Cref{alg:boundary-search}.

\begin{algorithm}
\caption{\sc Block-Decode}
\label{alg:boundary-search}
\textbf{Input:} An index $i \in [n']$ and oracle access to (corrupted) codeword $c' \in \zo^{n'}$.\\
\textbf{Output:} Some string $\Dec(s)$ for a substring $s$ of $c'$, or $\bot$.
\begin{algorithmic}[1]
\Function{Block-Decode$^{c'}$}{$i$}
    \State $\mathsf{buff} \gets \textproc{Buff-Find}_\eta^{c'}(i)$
    \If{$\mathsf{buff}==\bot$}
        \State \Return $\bot$
    \Else ~Parse $\mathsf{buff}$ as $(a,b), (a',b')$
        \If{$b < i < a'$}
            \State \Return $\Dec_{in}(c'[ b+1, a'-1 ])$
        \EndIf
    \EndIf
    \State\Return $\bot$
\EndFunction
\end{algorithmic}
\end{algorithm}

\begin{algorithm}[h]
\caption{\sc Buff-Find$_\eta$}
\label{alg:buff-find}
\textbf{Input:} An index $i \in [n']$ and oracle access to (corrupted) codeword $c' \in \set{0,1}^{n'}$.\\ 
\textbf{Output:} Two consecutive $\delta_\buff$-approximate buffers $(a,b), (a',b')$, or $\bot$.
\begin{algorithmic}[1]
\Function{Buff-Find$^{c'}$}{$i$}
	\State $j_s \gets \max\{1, i-\eta \tau\}$, $j_e \gets \min\{n'-\tau+1, i+\eta \tau\}$
	\State $\mathsf{buffs} \gets [~]$
	\While{$j_s \leq j_e$}
		\If{$\ED(0^\tau, c'[j_s, j_s+\tau-1]) \leq \delta_\buff \alpha\tau$}
			\State $\mathsf{buffs}.\mathsf{append}((j_s, j_s+\tau-1))$
		\EndIf
		\State $j_s \gets j_s + 1$
	\EndWhile
	\ForAll{$k \in \set{0,1,\dotsc, |\mathsf{buffs}|-2}$}
		\State $(a,b) \gets \mathsf{buffs}[k]$, $(a',b') \gets \mathsf{buffs}[k+1]$
		\If{$b < i < a'$}
			\State \Return $(a,b),(a',b')$
		\EndIf
	\EndFor
	\State \Return $\bot$
\EndFunction
\end{algorithmic}
\end{algorithm}

\subsection{\textsc{Buff-Find}}
The algorithm \bbs makes use of the sub-routine \buffind, presented in \cref{alg:buff-find}.
At a high-level, the algorithm \buffind on input $i$ and given oracle access to (corrupted) codeword $c'$ searches the ball $c'[i-\eta\tau,i+\eta\tau]$ for all {\em $\delta_\buff$-approximate buffers} in the interval, where $\eta \geq 1$ is a constant such that if $i \in \phi^{-1}(j)$ for any good block $j$ then $c'[ \phi^{-1}(j) ] \subseteq c'[i-\eta\tau, i+\eta\tau]$. 
Briefly, for any $k \in \bbN$ and $\delta_\buff \in (0,1/2)$ a string $w \in \zo^{k}$ is a {\em $\delta_\buff$-approximate buffer} if $\ED(w, 0^{k}) \leq \delta_\buff\cdot k$. 
For brevity we refer to approximate buffers simply as buffers.
Once all buffers are found, the algorithm attempts to find a pair of consecutive buffers such that the index $i$ is between these two buffers.
If two such buffers are found, then the algorithm returns these two consecutive buffers.
For notational convenience, for integers $a < b$ we let the tuple $(a,b)$ denote a (approximate) buffer.
\begin{lemma}\label{lem:good-approx-buffs}
    Let $i \in [n']$ and $j \in [d]$.
    There exist constants $\gamma < \delta_\buff \in (0,1/2)$ such that if $i \in \phi^{-1}(j)$ then \buffind finds buffers $(a_1, b_1)$ and $(a_2, b_2)$ such that $\Dec_{in}( c'[b_1 + 1, a_2 - 1] ) = j \circ b_j$.
    Further, if $b_1 < i < a_2$ then \bbs outputs $j \circ b_j$. 
\end{lemma}
\begin{proof}
	We first examine an uncorrupted block which has the form $B = 0^{\alpha\tau} \circ Y \circ 0^{\alpha\tau}$ for some $Y \in C_{in}$.
	Let $s = |B| = \beta\tau$ and note that $\tau = \log m$ and $|Y| = (\beta-2\alpha)\tau$.
	Note also that $B[1,\alpha\tau] = B[s,s-\alpha\tau+1] = 0^{\alpha\tau}$ and $B[\alpha\tau+1, s-\alpha\tau] = Y$.
	We observe that approximate buffers $(a,b)$ exist such that $b > \alpha\tau$ or $a < s-\alpha\tau+1$; that is, an approximate buffer can cut into the codeword $Y$.
	We are interested in bounding how large this ``cut'' can be in a $\gamma$-good block and first examine how large this ``cut'' can be in an uncorrupted block.
	
	Our inner code has the property that any interval of length $2\log (\alpha\tau)$ has at least fractional weight $\geq 2/5$.
	That is, an interval of length $2\log(\alpha\tau)$ in $Y$ has at least $(4/5)\log (\alpha\tau)$ number of $1$'s.
	Also any approximate buffer has weight at most $(\delta_\buff/2) \alpha\tau$.
	Let $\ell = c_0\tau$ for some constant $c_0$.
	We count the number of 1's in any $c_0\tau$ interval of $Y$.
	Note that in such an interval there are at most $c_0\tau/(2\log(\alpha\tau))$ disjoint intervals of length $2\log(\alpha\tau)$.
	Since the weight of each of these $2\log(\alpha\tau)$ intervals is at least $(4/5)\log(\alpha\tau)$ and the intervals are disjoint, we have that the weight of the interval $c_0\tau$ in $Y$ is at least $c_0\tau/(2\log(\alpha\tau)) \cdot (4/5)\log(\alpha\tau) = (2/5)c_0\tau$.
	We pick $c_0$ such that $(2/5)c_0\tau \geq (\delta_\buff/2) \alpha\tau + 1$; i.e., $c_0 = (5/4)\delta_\buff\alpha+1 \geq (5/4)\alpha\delta_\buff + 5/(2\tau)$ (for large enough $m$ since $\tau$ is an increasing function of $m$).
	On the other hand, we can have that an interval of length $(\delta_\buff/2)\alpha\tau+1$ in $Y$ has $(\delta_\buff/2)\alpha\tau+1$ number of 1's.
	
	The above derivation implies that largest ``cut'' an approximate buffer can make into the codeword $Y$ from the start (i.e., indicies after $\alpha\tau$) (and symmetrically the end; i.e., indices before $s-\alpha\tau+1$) has size in the range $[ (1+\delta_\buff/2)\alpha\tau, (1+5\delta_\buff/2)\alpha\tau]$.
	This implies that there exists $b_1,a_2 \in \bbN$ such that $(1+\delta_\buff/2)\alpha\tau \leq b_1 \leq (1+5\delta_\buff/2)\alpha\tau$ and $(\beta-\alpha(1+5\delta_\buff/2))\tau\leq a_2\leq (\beta - \alpha(1+\delta_\buff/2))\tau$.
	Further $b_1$ and $a_2$ have the following properties: (1) $B[b_1-\tau+1, b_1]$ and $B[a_2, a_2+\tau-1]$ are approximate buffers; and (2) for every $i \in \{ b_1-\tau+2, b_1-\tau+3,\dotsc, a_2-1 \}$, the window $B[i,i+\tau-1]$ is not an approximate buffer.
	These properties follow by our choice of $c_0$ and by the density property we have for our inner code $C_{in}$.
	
	We obtained the above bounds on $b_1,a_2$ by analyzing an uncorrupted block $B$.
	We use this as a starting point for analyzing a $\gamma$-good block $\tilde{B}$ (i.e., $\ED(B,\tilde{B}) \leq \gamma \alpha\tau$.
	Let $s' = |\tilde{B}|$.
	Then by $\gamma$-good we have that $(1-\alpha\gamma)s \leq s' \leq (1+\alpha\gamma)s$.
	Now by $\gamma$-good, we have that the bounds obtained on $b_1$ and $a_2$ are perturbed by at most $\alpha\gamma\tau$.
	That is, we have in block $\tilde{B}$
	\begin{align*}
		& & (1+\delta_\buff/2-\gamma)\alpha\tau &\leq b_1 \leq (1+5\delta_\buff/2+\gamma)\alpha\tau\\
		& & (\beta - \alpha(1 + 5\delta_\buff/2+\gamma))\tau &\leq a_2 \leq (\beta - \alpha(1+\delta_\buff/2-\gamma))\tau.
	\end{align*}
	This gives us
	\begin{align*}
		& & a_2 - b_1 &\leq (\beta - \alpha(1+\delta_\buff/2-\gamma))\tau - (1+\delta_\buff/2-\gamma)\alpha\tau\\
		& & &= (\beta -2\alpha(1+\delta_\buff/2-\gamma))\tau\\
		& & a_2 - b_1 &\geq (\beta - \alpha(1 + 5\delta_\buff/2+\gamma))\tau - (1+5\delta_\buff/2+\gamma)\alpha\tau\\
		& & &= (\beta - 2\alpha(1 + 5\delta_\buff/2+\gamma))\tau.
	\end{align*}
	Now we want to ensure decoding is possible on $c'[b_1+1, a_2-1]$.
	We observe that $(\beta-2\alpha)\tau - (a_2-b_1)$ is the number of insdels that are introduced because of the buffer finding algorithm.
	This quantity can be written as
	\begin{align*}
		& & (\delta_\buff-2\gamma)\alpha\tau \leq (\beta-2\alpha)\tau-(a_2-b_1) \leq (5\delta_\buff+2\gamma)\alpha\tau.
	\end{align*}
	Note that since $\abs{(\delta_\buff-2\gamma)\alpha\tau} \leq \abs{(5\delta_\buff+2\gamma)\alpha\tau}$, we can correctly decode if $\gamma$ and $\delta_\buff$ are chosen such that
	\begin{align*}
		& & (5\delta_\buff + 2\gamma)\alpha\tau + \gamma\alpha\tau &\leq \delta_\inn (\beta-2\alpha)\tau\\
		& & \frac{(5\delta_\buff+3\gamma)\alpha}{\beta-2\alpha} &\leq \delta_\inn.
	\end{align*}
	
	To finish, we note that the constant $\eta$ is chosen so that if $i \in \phi^{-1}(j)$ for any good block $j$ then we have that $c'[\phi^{-1}(j)] \subset c'[i-\eta\tau, i+\eta\tau]$.
	Since the algorithm \buffind finds every $\delta_\buff$-approximate buffer in the interval $c'[i-\kappa\tau, i+\kappa\tau]$ and since this interval contains $\gamma$-good block $j$, we have that the algorithm indeed \buffind returns approximate buffers $(a_1,b_1)$ and $(a_2,b_2)$ such that $\Dec_{in}(c'[b_1+1, a_2-1]) = j\circ Y\p{j})$ if $b_1 < i < a_2$, thus proving the lemma. 
\end{proof}

We now recall and prove \cref{thm:buffer-finding-properties}.

\buffProperties*
\begin{proof}
	First we analyze the probability $\Pr_{i \in \phi^{-1}(j)}[ \varphi(i) \neq j ].$
	By \Cref{lem:good-approx-buffs} the algorithm \bbs on input $i$ correctly outputs the block $Y\p{j} \circ j$ if $i \in [b_1+1, a_2-1] \subset \phi^{-1}(j)$.
	Since $j$ is $\gamma$-good and by \Cref{prop:bounds-good} we have that $|\phi^{-1}(j)| \leq (\beta+\alpha\gamma)\tau$.
	Finally, by correctness of the decoder $\Dec_{in}$, \cref{lem:good-approx-buffs} gives us a lower bound on the distance $a_2-b_1$.
	In particular,
	\begin{align*}
		& &a_2 - b_1 \geq (\beta - 2\alpha(1+5\delta_\buff/2 + \gamma))\tau.
	\end{align*}
	Thus we have that
	\begin{align*}
		& &\Pr_{i \in \phi^{-1}(j)}[ \varphi(i) \neq j ] &= 1- \Pr_{i \in \phi^{-1}(j)}[ \varphi(i) = j ] = 1 - \frac{a_2-b_1}{|\phi^{-1}(j)|} \leq 1 - \frac{a_2 - b_1}{(\beta+\alpha\gamma)\tau}\\
		& & &\leq 1 - \frac{(\beta - 2\alpha(1 + 5\delta_\buff/2+\gamma))\tau}{(\beta+\alpha\gamma)\tau} =1 - \frac{\beta - 2\alpha(1 + 5\delta_\buff/2+\gamma)}{(\beta+\alpha\gamma)} \\
		& & &\leq \frac{\alpha\gamma + 6\alpha}{\beta+\alpha\gamma} \leq \frac{\tp{\gamma+6}\alpha}{2} \le \gamma,
	\end{align*}
	where we assumed that that $\delta_\buff < 1/2$, $\gamma=1/12$ and $\alpha \le 2\gamma/(\gamma+6)$.
	More generally, there exists constants $\delta_\buff, \gamma$, and $\alpha$ such that the above inequalities hold with $\alpha \leq 2\gamma/(\gamma+6)$.

	For the second statement of \Cref{thm:buffer-finding-properties}, we analyze the algorithm \textproc{Interval-Decode}.
	Note we are only concerned with $\gamma$-good blocks which are wholly contained in the interval $[l,r)$. 
	Let $\+I[L,R-1]$ be the closure of $[l,r)$. 
	We note that $\phi$ restricted to $\+I[L,R-1]$ is a sub-decomposition which captures the errors introduced to blocks $L, \ldots, R-1$. 
	The algorithm \textproc{Interval-Decode} is similar to the global buffer-finding algorithm of SZ codes applied to the interval $[l,r)$: it searches intervals of length $\alpha\tau$ in $\{l,l+1,\dotsc, r-1\}$ from left to right until an approximate buffer $c'[i,i+\alpha\tau-1]$ is found. 
	Then the algorithm marks it and continue scanning for approximate buffers, starting with left endpoint of the first new interval at the right endpoint of the presumed buffer. 
	Then once the whole interval has been scanned, the algorithm finds pairs of consecutive buffers which are far apart and attempts to decode the section of the block that falls between these two buffers.
	
	According to the analysis of the SZ buffer finding algorithm, as long as block $j$ and $j+1$ are $\gamma$-good (for small enough constant $\gamma$), the buffers surrounding block $j+1$ should be located approximately correctly, and block $j$ will appear close to a codeword. 
	Since every block in the closure of $[l,r)$ is $\gamma$-good, all the buffers in this interval should be located approximately correctly, and every block $j$ such that $\phi^{-1}(j) \subseteq [l,r)$ should be decoded properly. 
	Therefore there will be exactly one block decoded to $(j, b)$ and it must hold that $b = b_j$.
	
	There is one minor issue with the above argument. 
	The searching process starts from an index $l$ which does not necessarily align with the left boundary of $\+I[L,R-1]$. 
	However, we note that this only affects the location of the first approximate buffer, and all subsequent buffers are going to be consistent with what the algorithm would have found if it started from the left boundary of $\+I[L,R-1]$. 
	In order to decode the first block, \textproc{Interval-Decode} performs another SZ buffer finding algorithm, but \emph{from right to left}, and decodes the leftmost block.
\end{proof}

\section{Parameter Setting and Proof of \texorpdfstring{\Cref{thm:main}}{Theorem 6}}\label{sec:param-and-proof}

In this section we list a set of constraints which our setting of parameters must satisfy, and then complete the proof of \Cref{thm:main}. 
These constraints are required by different parts of the analysis. 
Recall that $\delta_{out}, \delta_{in} \in (0,1)$ and $\beta_{in} \ge 1$ are given as parameters of the outer code and the inner code, and that $\beta = 2\alpha + \beta_{in}\tp{1 + \log|\Sigma|}$. We have that $\beta \ge 2$ for any non-negative $\alpha$.
\begin{proposition} \label{prop:parameters}
There exists constants $\gamma, \theta \in (0,1)$ and $\alpha = \Omega(\delta_{in})$ such that the following constraints hold:
\begin{enumerate}
    \item $\gamma \le 1/12$ and $\theta < 1/50$; 
    \item $(\beta+\gamma)/(\beta-\gamma) < 4/3$;
    \item $\alpha \le 2\gamma/(\gamma+6)$;
    \item $\alpha(1 + 3\gamma)/(\beta-2\alpha) < \delta_{in}$.
\end{enumerate}
\end{proposition}
\begin{proof}
For convenience of the reader and simplicity of the presentation 
we work with explicit values and verify that they satisfy the constraints in \cref{prop:parameters}.
Let $\gamma = 1/12$ and $\theta = 1/51$, which satisfies constraint (1). 
Note that $\gamma < 2/7 \le \beta/7$, hence 
$$ \frac{\beta+\gamma}{\beta-\gamma} < \frac{4}{3} $$
and constraint (2) is satisfied. 
We take $\alpha = 2\gamma\delta_{in}/(\gamma+6)$ so that $\alpha = \Omega(\delta_{in})$ and constraint (3) is satisfied. 
Note also that $\beta - 2\alpha = \beta_{in}(1 + \log|\Sigma|) \ge 2$ which implies
$$ \frac{\alpha(1 + 3\gamma)}{\beta-2\alpha} \le \frac{\alpha(1+3\gamma)}{2} = \frac{\alpha(\gamma + 3\gamma^2)}{2\gamma} < \frac{\alpha(\gamma+6)}{2\gamma} = \delta_{in}. $$
Therefore, constraint (4) is also satisfied.
\end{proof}

We let
$$ \delta = \frac{\delta_{out} \alpha\gamma}{2\beta(1+1/\theta)} = \Omega\tp{ \delta_{in}\delta_{out} }. $$
We now recall and prove \Cref{thm:main}, which shows \Cref{thm:general}.

\main*
\begin{proof}
Recall that the decoder $\Dec$ works as follows. 
Given input index $i \in [k]$ and oracle access to $c' \in \set{0,1}^{n'}$, $\Dec^{c'}(i)$ simulates $\Dec_{out}^{s'}(i)$. 
Whenever $\Dec_{out}^{s'}(i)$ queries an index $j \in [m]$, the decoder expresses $j = (p-1)\tau + r_j$ for $p \in [d]$ and $0\le r_j < \tau$, and runs \cref{alg:noisy-binary-search} on input $(c',1,n'+1,p)$ to  obtain a $\tau$-long string $b'_p$. 
Then it feeds the $(r_j+1)$-th symbol of $b'_p$ to $\Dec_{out}^{s'}(i)$. 
At the end of the simulation, $\Dec^{c'}(i)$ returns the output of $\Dec_{out}^{s'}(i)$.

For $p \in [d]$, let $b'_p \in \Sigma^{\tau} \cup \set{\perp}$ be a random variable that has the same distribution as the output of \cref{alg:noisy-binary-search} on input $(c',1,n'+1,p)$. 
Define a random string $s' \in \tp{\Sigma \cup \set{\perp}}^m$ as follows. 
For every $i \in [m]$ such that $i=(p-1)\tau + r$ for $p \in [d]$ and $0\le r < \tau$, 
\begin{align*}
& & s'[i] = \begin{cases}
b'_p[r] & \textup{if $b'_p \neq \perp$}, \\
\perp & \textup{if $b'_p = \perp$}.
\end{cases}
\end{align*}

Since $b'_p = b_p$ implies $s'[(p-1)\tau + r] = s[(p-1)\tau + r]$ for all $0 \le r < \tau$, the event $E_s \coloneqq \set{\Pr_{j \in [m]}\left[ s'[j] \neq s[j] \right] \le \delta_{out} }$ is implied by the event $E_b \coloneqq \{\Pr_{j \in [d]}\big[ b'_j \neq b_j \big] \le \delta_{out} \}$. 
\cref{thm:noisybs-to-outer-decoder} implies that $\Pr[E_s] \ge \Pr[E_b] \ge 1-\negl(n')$. 
According to the construction of $\Dec$, from the perspective of the outer decoder, the string $s'$ is precisely the string it is interacting with. 
Hence by properties of $\Dec_{out}$ we have that
\begin{align*}
& &\forall i \in [k], \quad \Pr\left[ \Dec_{out}^{s'}(i) = x[i] \ \middle| \ E_s \right] \ge \frac{1}{2} + \eps_{out}.
\end{align*}
Therefore by construction of $\Dec$ we have
\begin{align*}
& &\forall i \in [k], \quad \Pr\left[ \Dec^{c'}(i) = x[i] \right] &\ge \Pr\left[ E_s \right] \cdot \Pr\left[ \Dec_{out}^{s'}(i) = x[i] \ \middle| \ E_s \right] \\
& & &\ge \tp{1-\negl(n')} \cdot \tp{\frac{1}{2} + \eps_{out}} \ge \frac{1}{2} + \eps_{out} - \negl(n').
\end{align*}
The query complexity of $\Dec$ is $\ell_{out} \cdot O\tp{\log^4 n'}$ since it makes $\ell_{out}$ calls to \cref{alg:noisy-binary-search}, which by \cref{prop:noisybs-query-complexity} has query complexity $O\tp{\log^4 n'}$.
\end{proof}

\bibliography{references,abbrev0,crypto}

\newpage
\appendix
\section{Proof of \texorpdfstring{\cref{thm:main-noisy-binary-search}}{Theorem 19}}\label{app:main-noisybs-proof}
We first recall \cref{thm:main-noisy-binary-search}.

\mainNoisyBinarySearch*
We emphasize that the exact boundaries of any block $\phi^{-1}(j)$ or interval $\+I[L,R]$ are not known to the binary search algorithm, so it cannot do uniform sampling within the exact boundaries. Instead, as we can see in \cref{alg:noisy-binary-search}, in each iteration it picks two indices $m_1, m_2$ and calls \textsc{Block-Decode} on uniformly sampled indices in $\set{m_1, m_1+1, \cdots, m_2-1}$. Depending on the results returned by \textsc{Block-Decode}, it either sets $l=m_1$ or $r=m_2$ and recursively search in the smaller interval $[l, r)$.

The following lemma shows that as long as the closure of an interval $[l,r)$ is  $(\theta,\gamma)$-good, uniform samples from $[l,r)$ does now perform much worse than uniform samples from a good block in terms of estimating $\phi$.

\begin{lemma} \label{lem:boundary-oblivious-sampling}
	Let $[l,r)$ be an interval with closure $\+I[L,R-1]$. Suppose $\+I[L,R-1]$ is a $(\theta,\gamma)$-good interval. We have
	\begin{align*}
		& &\Pr_{i \in [l,r)}\left[ \varphi(i) \neq \phi(i) \right] \le \gamma + \theta + \frac{\gamma}{\beta}.
	\end{align*}
\end{lemma}
\begin{proof}
	Let $\Goody \subseteq \set{L+1, \ldots R-2}$ be the set of $\gamma$-good blocks among $\set{L+1, \ldots, R-2}$, and let $\overline{\Goody} = \set{L+1,\ldots,R-2} \setminus \Goody$. By definition of $(\theta,\gamma)$-goodness we have $\overline{\Goody} \le \theta(R-L)$. Since for each $j \in \Goody$, $\phi^{-1}(j) \subseteq [l,r)$. We can apply item (1) of \cref{thm:buffer-finding-properties} and get
	\begin{align*}
		& &\Pr_{i \in [l,r)}\left[ \varphi(i) \neq \phi(i) \mid \phi(i) \in \Goody \right] \le \gamma.
	\end{align*}
	Now we have the bound
	\begin{align*}
		& &\Pr_{i \in [l,r)}\left[ \varphi(i) \neq \phi(i) \right] &\le \Pr_{i \in [l,r)}\left[ \varphi(i) \neq \phi(i) \mid \phi(i) \in \Goody \right] + \Pr_{i \in [l,r)}\left[ \phi(i) \notin \Goody \right] \\
		& & &\le \gamma + \Pr_{i \in [l,r)}\left[ \phi(i) \notin \Goody \right],
	\end{align*}
	so it suffices to upper bound $\Pr_{i \in [l,r)}\left[ \phi(i) \notin \Goody \right]$. Denote $\Delta_j = \abs{\phi^{-1}(j)}-\beta\tau$. It holds that $\sum_{j=L}^{R-1}\abs{\Delta_j} \le \gamma(R-L)\tau$. In particular $\sum_{j \in \Goody}\Delta_j \ge -\gamma(R-L)\tau$ and $\sum_{j \in \overline{\Goody}}\Delta_j \le \gamma(R-L)\tau$. We have
	\begin{align*}
		\Pr_{i \in [l,r)}[\phi(i) \notin \Goody] &\le \frac{\sum_{j \notin \Goody}\abs{\phi^{-1}(j)}}{r-l} \le  \frac{\theta(R-L)\beta\tau + \sum_{j \notin \Goody}\Delta_j}{(R-L)\beta\tau + \sum_{j \in \Goody}\Delta_j + \sum_{j \notin \Goody}\Delta_j} \\
		&\le \frac{\theta(R-L)\beta\tau + \gamma(R-L)\tau}{(R-L)\beta\tau} = \theta + \frac{\gamma}{\beta}.
	\end{align*}
	Hence the lemma follows.
\end{proof}

In the following, we set $\rho = \min\set{\frac{1}{4} \cdot \frac{\beta-\gamma}{\beta+\gamma}, 1-\frac{3}{4}\cdot\frac{\beta+\gamma}{\beta-\gamma}}$ as in \cref{alg:noisy-binary-search}. Note that by item (2) of \cref{prop:parameters} we have $\rho > 0$.

The following lemma states that any interval not too far from a locally good block is also good. 
\begin{lemma}\label{lem:proximity-good-interval}
Let $l, r \in [n']$ be such that $r-l \ge 18(\beta+\gamma)\tau$. Let $\+I[L, R-1]$ be the closure of $[l, r)$. Set $m_1 = (1-\rho) l + \rho r$ and $m_2 = \rho l + (1-\rho) r$ and let $\+I[M_1, M_2-1]$ be the closure of $[m_1, m_2)$. Suppose for some $L \le x \le M_1$ block $x$ is $(\theta,\gamma)$-locally-good. Then we have
\begin{enumerate}
\item $M_1 \le L+(R-L)/3$, $M_2 \ge L+2(R-L)/3$.
\item $\+I[M_1, M_2-1]$ is a $(2\theta, 2\gamma)$-good interval.
\end{enumerate}
\end{lemma}
\begin{proof}
Since $L \le x \le R-1$ and block $x$ is $(\theta,\gamma)$-locally good, by definition $\+I[L,R-1]$ is a $(\theta,\gamma)$-good interval. From the inclusion $[l,r) \subseteq \+I[L,R-1]$ we know that
\begin{align*}
& &(R-L)(\beta+\alpha\gamma)\tau \ge \abs{\+I(L, R-1)} \ge r-l \ge 18(\beta+\alpha\gamma)\tau,    
\end{align*}
which implies $R - L \ge 18$.

We begin by proving item (1).

\begin{claim}
$M_1 \le L + (R-L)/3$.
\end{claim}
\begin{proof}
Suppose $M_1 > L + (R-L)/3$. From the inclusion $ \+I[L+1, M_1-1] \subseteq [l, m_1)$, we have
\begin{align*}
 & &\rho (r-l) = m_1 - l &\ge \abs{\+I[L+1, M_1-1]} \ge (M_1-L-1)(\beta-\alpha\gamma)\tau\\
 & & &> \frac{1}{4}\tp{R-L} \cdot (\beta-\alpha\gamma)\tau.
\end{align*}

The last inequality holds as long as $R-L \ge 12$. Similarly, from the inclusion $[l,r) \subseteq \+I[L,R-1]$, we have that
\begin{align*}
& &r-l \le |\+I[L,R-1]| \le (R-L)(\beta+\alpha\gamma)\tau. 
\end{align*}

This implies $\rho > \frac{1}{4} \cdot \frac{\beta-\gamma}{\beta+\gamma}$ which is a contradiction.
\end{proof}

\begin{claim}
$M_2 \ge L+2(R-L)/3$.
\end{claim}
\begin{proof}
Suppose $M_2 < L + 2(R-L)/3$. From the inclusion $[l, m_2) \subseteq \+I[L, M_2-1]$, we have
\begin{align*}
	& &(1-\rho)(r-l) &= m_2 - l \le \abs{\+I[L, M_2-1]} \le (M_2-L)(\beta+\alpha\gamma)\tau\\
	& & &< \frac{3}{4}\tp{R-L-2} \cdot (\beta+\alpha\gamma)\tau.
\end{align*}

The last inequality holds as long as $R-L \ge 18$. Similarly, from the inclusion $\+I[L+1,R-2] \subseteq [l,r)$, we have that
\begin{align*}
& &r-l \ge |\+I[L+1,R-2]| \ge (R-L-2)(\beta-\alpha\gamma)\tau. 
\end{align*}

This implies $1-\rho < \frac{3}{4} \cdot \frac{\beta+\gamma}{\beta-\gamma}$ which is a contradiction.
\end{proof}

An immediate consequence of item (1) is that $M_2 - L \le 2(M_2 - M_1)$. Therefore, by $(\theta, \gamma)$-locally-goodness of $x$, we have
\begin{align*}
& &\sum_{j=M_1}^{M_2-1}\ED\tp{c'[\phi^{-1}(j)], \tilde{Y}\itn{j}} &\le \sum_{j=x}^{M_2-1}\ED\tp{c'[\phi^{-1}(j)], \tilde{Y}\itn{j}} \\
& & &\le \gamma \cdot (M_2-L)\tau \\
& & &\le 2\gamma \cdot (M_2-M_1)\tau.
\end{align*} 
Similarly, the number of $2\gamma$-bad blocks among $\set{M_1, \cdots, M_2-1}$ is at most the number of $\gamma$-bad blocks among $\set{L, \cdots, M_2-1}$, which is upper bounded by $\theta (M_2 - L) \le 2\theta (M_2 - M_1)$. Therefore the interval $\+I[M_1, M_2-1]$ is $(2\theta, 2\gamma)$-good.
\end{proof}

The following is the main lemma we use to prove \cref{thm:main-noisy-binary-search}.
\begin{lemma} \label{lem:main-invariant}
Assume $j \in [d]$ is a $(\theta, \gamma)$-locally-good block. Denote by $l\itn{t}$ , $r\itn{t}$ the values of $l$, $r$ at beginning of the $t$-th iteration when running \cref{alg:noisy-binary-search} on input $(c',1,n'+1,j)$. Suppose $r\itn{t}-l\itn{t} \ge 36(\beta+\gamma)\tau$. Then we have
$$\Pr\left[ \phi^{-1}(j) \subseteq \left[l\itn{t+1}, r\itn{t+1}\right) \ \middle| \ \phi^{-1}(j) \subseteq \left[l\itn{t}, r\itn{t}\right) \right] \ge 1 - \negl(n'),$$
where the probability is taken over the randomness of the algorithm.
\end{lemma} 
\begin{proof}
Let $m_1$ and $m_2$ be defined as in \cref{alg:noisy-binary-search}. Let $\+I[L, R-1]$ be closure of $[l\itn{t}, r\itn{t}]$, and let $\+I[M_1, M_2-1]$ be the closure of $[m_1, m_2)$. Since we always have $[m_1, m_2) \subseteq [l\itn{t+1}, r\itn{t+1})$, $\phi^{-1}(j) \subseteq [m_1, m_2)$ would immediately imply $\phi^{-1}(j) \subseteq [l\itn{t+1}, r\itn{t+1})$. In the rest of the proof, we assume $\phi^{-1}(j) \not\subseteq [m_1, m_2)$, which means $L \le j \le M_1$ or $M_2-1 \le j \le R$.

We may assume that $L \le j \le M_1$ since the other case $M_2-1 \le j \le R$ is completely symmetric. The condition $r\itn{t} - l\itn{t} \ge 36(\beta+\gamma)\tau$ implies that $R-L \ge 36$, and that $m_2 - m_1 \ge \tp{r\itn{t} - l\itn{t}}/2 \ge 18(\beta+\gamma)\tau$. Therefore we can apply \cref{lem:proximity-good-interval} to $m_1$, $m_2$ and get that (1) $M_2 - M_1 \ge \tp{R - L}/3 \ge 12$, and (2) $\+I[M_1, M_2-1]$ is a $(2\theta,2\gamma)$-good interval. Since $\+I[M_1, M_2-1]$ is the closure of $[m_1, m_2)$, \Cref{lem:boundary-oblivious-sampling} gives 
\begin{align*}
    & &\Pr_{i \in [m_1,m_2)}\left[ \varphi(i) \neq \phi(i) \right] \le 2\gamma + 2\theta + \frac{2\gamma}{\beta} < \frac{1}{4} + \frac{1}{25} < \frac{1}{3},
\end{align*}
where the second last inequality is because $\theta < 1/50$, $\gamma \le 1/12$ (i.e., item (1) of \cref{prop:parameters}) and $\beta \ge 2$.

Let $i_1, i_2, \cdots, i_N$ be the samples drawn by \cref{alg:noisy-binary-search}, which are independent and uniform samples from $[m_1, m_2)$. Define $X_j$ to be the indicator random variable of the event $\set{\varphi(i_j) = \perp} \cup \set{\varphi(i_j) < x}$, and define $Y_j$ to be the indicator random variable of the event $\set{\varphi(i_j) \neq \phi(i_j)}$. It follows that $\E[Y_j] < 1/3$, and $\phi^{-1}(x) \not\subseteq \left[l\itn{t+1}, r\itn{t+1}\right)$ if and only if $\sum_{j=1}^{N}X_j \ge N/2$. Therefore it suffices to upper bound the probability of the latter event.

We observe that if $i \in [m_1, m_2) \subseteq \+I[M_1, M_2]$, then $\phi(i) \ge M_1 \ge j$. Therefore $\varphi(i) = \phi(i)$ implies $\varphi(i) \ge j$, or in other words $X_j \le Y_j$. An application of Chernoff bound gives 
\begin{align*}
& &\Pr\left[ \sum_{j=1}^{N}X_j \ge \frac{N}{2} \right] &\le \Pr\left[ \sum_{j=1}^{N}Y_j \ge \frac{N}{2} \right] \le \Pr\left[ \sum_{j=1}^{N}Y_j \ge \tp{1+\frac{1}{2}}\sum_{j=1}^{N}\E[Y_j] \right] \\
& & &\le \exp\tp{-\frac{N}{36}}.
\end{align*}
Taking $N = \Theta(\log^2 n')$ gives $\Pr\left[ \sum_{j=1}^{N}X_j \ge \frac{N}{2} \right] \le \exp\tp{-\Theta\tp{\log^2 n'}} = \negl(n')$.
\end{proof}

We are now ready to prove \cref{thm:main-noisy-binary-search}.
\begin{proof}[Proof of \cref{thm:main-noisy-binary-search}]
Let $C = 36(\beta+\gamma)\tau$ be defined as in \cref{alg:noisy-binary-search}, and $T = O(\log n')$ be the number of iterations until $r-l \le C$. Denote by $l\itn{t}$ , $r\itn{t}$ the values of $l$, $r$ at beginning of the $t$-th iteration. Let $\*b$ be the random variable denoting the output of \cref{alg:noisy-binary-search}. We have 
\begin{align} \label{ineq:success-prob}
    & &\Pr[\*b = b_j] \ge \Pr\left[ \phi^{-1}(j) \subseteq [l\itn{T}, r\itn{T}) \right] \cdot \Pr\left[ \*b = b_j \ \middle| \ \phi^{-1}(j) \subseteq [l\itn{T}, r\itn{T}) \right].
\end{align}
We are going to lower bound both probabilities in the right-hand-side of (\cref{ineq:success-prob}).

According to the algorithm, we have the following inclusion chain 
\begin{align*}
& &[l\itn{T}, r\itn{T}) \subseteq [l\itn{T-1}, r\itn{T-1}) \subseteq \cdots \subseteq [l\itn{1}, r\itn{1}),    
\end{align*}
where $l\itn{1}=1$ and $r\itn{1}=n'+1$. By \cref{lem:main-invariant}, it holds that
\begin{align*}
\Pr\left[\phi^{-1}(j) \subseteq [l\itn{T}, r\itn{T}) \right] &= \prod_{t=1}^{T-1}\Pr\left[ \phi^{-1}(j) \subseteq \left[l\itn{t+1}, r\itn{t+1}\right) \ \middle| \ \phi^{-1}(j) \subseteq \left[l\itn{t}, r\itn{t}\right) \right] \\
&\ge \tp{1 - \negl(n')}^T  \ge 1 - T \cdot \negl(n') = 1 - \negl(n').
\end{align*}

For the second term in \cref{ineq:success-prob}, let $\+I[L,R-1]$ be the closure of $[l\itn{T}, r\itn{T})$. Then $R-L \le 2+36(\beta+\gamma)/(\beta-\gamma) \le 50$. Conditioned on $\phi^{-1}(j) \subseteq [l\itn{T}, r\itn{T}) \subseteq \+I[L,R-1]$, the interval $\+I[L,R-1]$ is $(\theta,\gamma)$-good. Therefore every block in $\+I[L,R-1]$ is $\gamma$-good since the number of $\gamma$-bad blocks is bounded by $(R-L)\theta \le 50\theta < 1$. Due to item (2) of \cref{thm:buffer-finding-properties}, we have 
\begin{align*}
& &\Pr\left[ \*b = b_j \ \middle| \ \phi^{-1}(j) \subseteq [l\itn{T}, r\itn{T}) \right] = 1.
\end{align*}
Hence the theorem follows.
\end{proof}

\end{document}